\RequirePackage{amsmath}
\documentclass[runningheads]{llncs}
\usepackage[T1]{fontenc}
\usepackage{graphicx}

\usepackage{amssymb,latexsym}
\usepackage{xspace}
\usepackage{hyperref}
\usepackage{url}
\usepackage{color}

\usepackage{etoolbox}
\AtEndEnvironment{proof}{\phantom{}\qed} 

\usepackage{multirow}
\usepackage[capitalise]{cleveref}
\usepackage{enumitem} 
\usepackage{xcolor} 

\usepackage{pifont}

\usepackage{tikz}
\usetikzlibrary{arrows,automata}
\usetikzlibrary{positioning}
\usetikzlibrary{shapes,shapes.geometric,fit,calc,automata,}

\newcommand{\Subject}[1]{\paragraph*{#1.}}
\newcommand{\St}{~|~}

\newcommand{\con}{\cdot}

\newcommand{\tuple}[1]{\langle #1  \rangle}

\newcommand{\ST}{~{\big |}\:}

\newcommand{\Nat}{\ensuremath{\mathbb{N}}\xspace}
\newcommand{\Rat}{\ensuremath{\mathbb{Q}}\xspace}

\newcommand{\A}{{\mathcal{A}}}
\newcommand{\B}{{\mathcal{B}}}
\ifdefined\C
\renewcommand{\C}{{\mathcal{C}}}
\else
\newcommand{\C}{{\mathcal{C}}}
\fi
\newcommand{\D}{{\mathcal{D}}}

\newcommand{\M}{{\mathcal{M}}}

\newcommand{\N}{{\mathcal{N}}}

\newcommand{\emptyword}{\varepsilon}

\renewcommand{\phi}{\varphi}

\usepackage{etoolbox}

\newcommand{\Omit}[1]{}

\newcommand{\inc}{\mbox{\sc inc}\xspace}
\newcommand{\dec}{\mbox{\sc dec}\xspace}
\newcommand{\goto}{\mbox{\sc goto }}
\newcommand{\halt}{\mbox{\sc halt}\xspace}
\newcommand{\jz}[3]{\mbox{\sc if $#1$=0 goto $#2$ else goto $#3$}\xspace}
\newcommand{\CMrun}{\psi}
\newcommand{\numof}[2]{\#(#1,#2)}
\newcommand{\qfr}{q_{\mathsf{freeze}}}
\newcommand{\qhalt}{q_{\mathsf{halt}}}
\newcommand{\qhc}{q_{\mathsf{HC}}}
\newcommand{\qzc}{q_{\mathsf{ZC}}^c}
\newcommand{\qpc}[1]{q_\mathsf{PC#1}^c}
\newcommand{\incdec}{\Sigma^{\inc\dec}}
\newcommand{\allgoto}{\Sigma^{\mbox{\sc goto}}}
\newcommand{\withouthalt}{\Sigma^{\mbox{\sc nohalt}}}

\newcommand{\minval}{\mbox{\sc minval}}
\newcommand{\maxval}{\mbox{\sc maxval}}
\newcommand{\lowest}{\mbox{\sc lowrun}}
\newcommand{\highest}{\mbox{\sc highrun}}
\newcommand{\lowestw}{\mbox{\sc lowword}}
\newcommand{\highestw}{\mbox{\sc highword}}
\newcommand{\maxdiff}[1]{\mbox{\sc maxdiff}(\mathsf{#1})}

\newcommand{\fin}{\tau}

\newcommand{\pref}{\mbox{\sc pref}\xspace}
\newcommand{\preferred}{\delta_{pr}}

\usepackage{braket} 

\newcommand{\comment}[1]{}
\newcommand{\checker}[2]{\vspace{5pt}\noindent{\bf #1 Checker#2.}}
\newcommand{\lchecker}[2]{\vspace{5pt}\noindent{\it #1 Checker#2.}}

\newtoggle{completeProofs}
\toggletrue{completeProofs}

\begin{document}
\title{On the Comparison of Discounted-Sum Automata with Multiple Discount Factors
\iftoggle{completeProofs}{\thanks{This is the full version of a chapter with the same title that appears in the FoSSaCS	2023 conference proceedings \cite{FoSSaCS2023}.}}{}
}
\titlerunning{Comparison of Discounted-Sum Automata with Multiple Discount Factors}
%
\author{Udi Boker\inst{}\thanks{Research supported by the Israel Science Foundation grant 2410/22.}\orcidID{0000-0003-4322-8892} 
	\and
Guy Hefetz\inst{}\orcidID{0000-0002-4451-6581}}
\authorrunning{U. Boker and G. Hefetz}
%
\institute{Reichman University, Herzliya, Israel\\ \email{udiboker@runi.ac.il}, \email{ghefetz@gmail.com}}
\maketitle              
\begin{abstract}
We look into the problems of comparing nondeterministic discounted-sum automata on finite and infinite words. That is, the problems of checking for automata $\A$ and $\B$ whether or not it holds that for all words $w$, $\A(w)=\B(w), \A(w)\leq\B(w)$, or $\A(w)<\B(w)$.

These problems are known to be decidable when both automata have the same single integral discount factor, while decidability is open in all other settings: when the single discount factor is a non-integral rational; when each automaton can have multiple discount factors; and even when each has a single integral discount factor, but the two are different.

We show that it is undecidable to compare discounted-sum automata with multiple discount factors, even if all are integrals, while it is decidable to compare them if each has a single, possibly different, integral discount factor.
To this end, we also provide algorithms to check for given nondeterministic automaton $\N$ and deterministic automaton $\D$, each with a single, possibly different, rational discount factor, whether or not $\N(w) = \D(w)$, $\N(w) \geq \D(w)$, or $\N(w) > \D(w)$ for all words $w$.

\end{abstract}
%
\section{Introduction}
Equivalence and containment checks of Boolean automata, namely the checks of whether $L(\A)=L(\B)$, $L(\A)\subseteq L(\B)$, or $L(\A)\subset L(\B)$, where $L(\A)$ and $L(\B)$ are the languages that $\A$ and $\B$ recognize, are central in the usage of automata theory in diverse areas, and in particular in formal verification (e.g,\ \cite{Var87,HTKB92,CDK93,THB95,Var96,KVW00}). 
Likewise, comparison of quantitative automata, which extends the equivalence and containment checks by asking whether $\A(w)=\B(w)$, whether $\A(w)\leq\B(w)$, or whether $\A(w)<\B(w)$ for all words $w$, are essential for harnessing quantitative-automata theory to the service of diverse fields and in particular to the service of quantitative formal verification (e.g,\ \cite{CDH10,ExpressivenessQuantitativeLanguages,QuantitativeLanguagesDefinedByFunctionalAutomata,bren16,HPR17,SKRV17,ComparatorAutomataInQuantitativeVerification,FLW20}). 

Discounted summation is a common valuation function in quantitative automata theory (e.g,\ \cite{Skew,AlternatingWeightedAutomata,ExpressivenessQuantitativeLanguages,CDH10}), as well as in various other computational models, such as games (e.g.,\  \cite{ZP96,Andersson06,DiscountingInSystems}), Markov decision processes (e.g,\ \cite{DiscountedMarkov,DiscountedDeterministicMarkov,MultiObjectiveDiscountedReward}), and reinforcement learning (e.g,\ \cite{IntroductionToReinforcementLearning,CoRR19}), as it formalizes the concept that an immediate reward is better than a potential one in the far future, as well as that a potential problem (such as a bug in a reactive system) in the far future is less troubling than a current one.

A nondeterministic discounted-sum automaton (NDA) has rational weights on the transitions, and a fixed rational discount factor $\lambda > 1$. 
The value of a (finite or infinite) run is the discounted summation of the weights on the transitions, such that the weight in the $i$th transition of the run is divided by $\lambda^i$. 
The value of a (finite or infinite) word is the infimum value of the automaton runs on it.
An NDA thus realizes a function from words to real numbers.

NDAs cannot always be determinized \cite{CDH10}, they are not closed under basic algebraic operations \cite{BH14}, and their comparison is not known to be decidable, relating to various longstanding open problems \cite{TDS}.
However, restricting NDAs to have an integral discount factor $\lambda\in\Nat\setminus\{0,1\}$ provides a robust class of automata that is closed under determinization and under algebraic operations, and for which comparison is decidable \cite{BH14}.

Various variants of NDAs are studied in the literature, among which are \emph{functional}, \emph{k-valued}, \emph{probabilistic}, and more \cite{QuantitativeLanguagesDefinedByFunctionalAutomata,FiniteValuedWeightedAutomata,ProbabilisticWeightedAutomata}.
Yet, until recently, all of these models were restricted to have a single discount factor.
This is a significant restriction of the general discounted-summation paradigm, in which multiple discount factors are considered. For example, Markov decision processes and discounted-sum games allow multiple discount factors within the same entity \cite{DiscountedMarkov,Andersson06}.
In \cite{BH21}, NDAs were extended to NMDAs, allowing for multiple discount factors, where each transition can have a different one. Special attention was given to integral NMDAs, namely to those with only integral discount factors, analyzing whether they preserve the good properties of integral NDAs. It was shown that they are generally not closed under determinization and under algebraic operations, while a restricted class of them, named tidy-NMDAs, in which the choice of discount factors depends on the prefix of the word read so far, does preserve the good properties of integral NDAs.

While comparison of tidy-NMDAs with the same choice function is decidable in PSPACE \cite{BH21}, it was left open whether comparison of general integral NMDAs $\A$ and $\B$ is decidable. It is even open whether comparison of two integral NDAs with different (single) discount factors is decidable.

We show that it is undecidable to resolve for given NMDA $\N$ and deterministic NMDA (DMDA) $\D$, even if both have only integral discount factors, on both finite and infinite words, whether $\N \equiv \D$ and whether $\N \leq \D$, and on finite words also whether $\N < \D$.
We prove the undecidability result by reduction from the halting problem of two-counter machines.
The general scheme follows similar reductions, such as in \cite{DDGRT10,ABK22}, yet the crux is in simulating a counter by integral NMDAs. 
Upfront, discounted summation is not suitable for simulating counters, since a current increment has, in the discounted setting,  a much higher influence than of a far-away decrement. However, we show that multiple discount factors allow in a sense to eliminate the influence of time, having automata in which no matter where a letter appears in the word, it will have the same influence on the automaton value. (See \cref{lem:undecidabilityContainment,fig:undecidabilityContainment_A}). Another main part of the proof is in showing how to nondeterministically adjust the automaton weights and discount factors in order to ``detect'' whether a counter is at a current value $0$. (See \cref{fig:NegativeCounter,fig:balancedCounters,fig:zeroJumpChecker,fig:positiveJumpChecker}.)

On the positive side, we provide algorithms to decide for given NDA $\N$ and deterministic NDA (DDA) $\D$, with arbitrary, possibly different, rational discount factors, whether $\N \equiv \D$, $\N \geq \D$, or $\N > \D$ (\cref{cl:TwoNdasDecision}).
Our algorithms work on both finite and infinite words, and run in PSPACE when the automata weights are represented in binary and their discount factors in unary.
Since integral NDAs can always be determinized \cite{BH14}, our method also provides an algorithm to compare two integral NDAs, though not necessarily in PSPACE, since determinization might exponentially increase the number of states. (Even though determinization of NDAs is in PSPACE \cite{BH14,BH21}, the exponential number of states might require an exponential space in our algorithms of comparing NDAs with different discount factors.)

The challenge with comparing automata with different discount factors comes from the combination of their different accumulations, which tends to be intractable, resulting in the undecidability of comparing integral NMDAs, and in the open problems of comparing rational NDAs and of analyzing the representation of numbers in a non-integral basis \cite{Mah68,GS01,Har06,TDS}.
Yet, the main observation underlying our algorithm is that when each automaton has a single discount factor, we may unfold the combination of their computation trees only up to some level $k$, after which we can analyze their continuation separately, first handling the automaton with the lower (slower decreasing) discount factor and then the other one.
The idea is that after level $k$, since the accumulated discounting of the second automaton is already much more significant, even a single non-optimal transition of the first automaton cannot be compensated by a continuation that is better with respect to the second automaton. We thus compute the optimal suffix words and runs of the first automaton from level $k$, on top which we compute the optimal runs of the second automaton.

\section{Preliminaries}
\Subject{Words}
An \emph{alphabet} $\Sigma$ is an arbitrary finite set, and a \emph{word} over $\Sigma$ is a finite or infinite sequence of letters in $\Sigma$, with $\emptyword$ for the empty word. We denote the concatenation of a finite word $u$ and a finite or infinite word $w$ by $u\con w$, or simply by $uw$.
We define $\Sigma^+$ to be the set of all finite words except the empty word, i.e., $\Sigma^+=\Sigma^*\setminus\{\emptyword\}$.
For a word $w=\sigma_0 \sigma_1 \sigma_2 \cdots$ and indexes $i\leq j$, we denote the \emph{letter at index $i$} as $w[i]=\sigma_i$, and the \emph{sub-word from $i$ to $j$} as $w[i..j]=\sigma_i \sigma_{i+1} \cdots \sigma_j$.

For a finite word $w$ and letter $\sigma\in\Sigma$, we denote the number of occurrences of $\sigma$ in $w$ by $\numof{\sigma}{w}$, and for a set $S\subseteq\Sigma$, we denote $\sum_{\sigma\in S}\numof{\sigma}{w}$ by $\numof{S}{w}$.

For a finite or infinite word $w$ and a letter $\sigma\in\Sigma$, we define the \emph{prefix of $w$ up to $\sigma$}, $\pref_\sigma(w)$, as the minimal prefix of $w$ that contains a $\sigma$ letter if there is a $\sigma$ letter in $w$ or $w$ itself if it does not contain any $\sigma$ letters. Formally,\
$\pref_\sigma(w) =\begin{cases}
	w\big[0..\min \{i \St w[i] = \sigma\}\big] & \exists i \St w[i]=\sigma\\
	w & \text{otherwise}
\end{cases} $

\Subject{Automata}
A nondeterministic discounted-sum automaton (NDA) \cite{CDH10} is an automaton with rational weights on the transitions, and a fixed rational discount factor $\lambda > 1$. 
A nondeterministic discounted-sum automaton with multiple discount factors (NMDA) \cite{BH21} is similar to an NDA, but with possibly a different discount factor on each of its transitions. They are formally defined as follows:

\begin{definition} [{\cite{BH21}}]\label{def:NMDA}
	A nondeterministic discounted-sum automaton with multiple discount factors (NMDA), on finite or infinite words, is a tuple $\A = \tuple{\Sigma, Q, \iota, \delta, \gamma, \rho}$ over an alphabet $\Sigma$, with a finite set of states $Q$, an initial set of states $\iota\subseteq Q$, a transition function $\delta \subseteq Q \times \Sigma \times Q$, a weight function $\gamma: \delta\to\Rat$, 
	and a discount-factor function $\rho: \delta \to \Rat\cap (1,\infty)$, assigning to each transition its discount factor, which is a rational greater than one.
	\footnote{Discount factors are sometimes defined as numbers between $0$ and $1$, under which setting weights are multiplied by these factors rather than divided by them.}
	\begin{itemize}
		\item 
		A \emph{run} of $\A$ is a sequence of states and alphabet letters, $p_0, \sigma_0, p_1, \sigma_1, p_2, \cdots$, such that $p_0\in\iota$ is an initial state, and for every $i$, 
		$(p_i,\sigma_i,p_{i+1})\in\delta$.
		
		\item 
		The \emph{length} of a run $r$, denoted by $|r|$, is $n$ for a finite run $r = p_0, \sigma_0, p_1, \allowbreak\cdots,
		\sigma_{n-1}, p_n$, and $\infty$ for an infinite run.
		
		\item 
		For an index $i<|r|$, we define the $i$-th transition of $r$ as $r[i]=(p_{i},\sigma_i,p_{i+1})$, and the prefix run with $i$ transitions as $r[0..i]=p_0, \sigma_0, p_1, \cdots, \sigma_i, p_{i+1}$.
		
		\item The \emph{value} of a finite/infinite run $r$ is
		$\A(r)=\sum_{i=0}^{|r|-1}{ \bigg(\gamma\big(r[i])\big) \cdot \prod_{j=0}^{i-1} \frac{1}{\rho\big(r[j]\big)}\bigg)}$.
		For example, the value of the run $r_1=q_0,a,q_0,a,q_1,b,q_2$ of $\A$ from \cref{fig:NMDAExample} is 
		$\A(r_1)= 1 + \frac{1}{2}\cdot\frac{1}{3} + 2\cdot\frac{1}{2\cdot 3}=\frac{3}{2}$.
		
		\item The \emph{value} of $\A$ on a finite or infinite word $w$ is\\
		$\A(w) = \inf \{\A(r) \St r \text{ is a run of } \A \text{ on } w \}$.
		
		\item
		For every finite run $r = p_0, \sigma_0, p_1, \cdots, \sigma_{n-1}, p_n$, we define
		the \emph{target state} as $\delta(r)=p_n$ and the \emph{accumulated discount factor} as $\rho(r)=\prod_{i=0}^{n-1}{\rho\big(r[i])\big)}$.
		
		\item
		When all discount factors are integers, we say that $\A$ is an \emph{integral} NMDA.
		
		\item In the case where $|\iota|= 1$ and for every $q \in Q$ and $\sigma \in \Sigma$, we have $|\{ q' \ST (q,\sigma,q')\in\delta \}| \leq 1$, we say that $\A$ is {\em deterministic}, denoted by DMDA, and view $\delta$ as a function from words to states.
		
		\item
		When the discount factor function $\rho$ is constant, $\rho\equiv\lambda\in\Rat\cap (1,\infty)$, we say that $\A$ is a \emph{nondeterministic discounted-sum automaton }(NDA) \cite{CDH10} with discount factor $\lambda$ (a  $\lambda$-NDA). If $\A$ is deterministic, it is a $\lambda$-DDA.
		
		\item
		For a state $q\in Q$, we write $\A^q$ for the NMDA $\A^q = \tuple{\Sigma, Q, \set{q}, \delta, \gamma, \rho}$.
		
	\end{itemize}
\end{definition}


\begin{figure}
	\vspace*{-\baselineskip}
	\centering
	\setlength{\belowcaptionskip}{-\baselineskip}
	\begin{tikzpicture}[->,>=stealth',shorten >=1pt,auto,node distance=2cm, semithick, initial text=, every initial by arrow/.style={|->}]
		\node [midway] [ left of=q0] {$\A:$};
		\node[initial, state] (q0) {$q_0$};
		\node[state] (q1) [right of=q0] {$q_1$};
		\node[state] (q2) [right of=q1] {$q_2$};
		
		\path 
		(q0) edge	[loop above, out=120, in=70,looseness=4] node [left, xshift=-0.2cm, yshift=-0.2cm] {$a,1,3$} (q0)
		(q1) edge	[loop above, out=120,in=70,looseness=4] node [right, xshift=0.2cm, yshift=-0.1cm] {$a,\frac{1}{2},2$} (q1)
		(q2) edge	[loop above, out=80, in=30, looseness=4] node [right, xshift=0.2cm, yshift=-0.2cm]{$a,\frac{1}{4},2$} (q2)
		(q2) edge	[loop below, out=-80, in=-30,looseness=4] node [right, xshift=0.2cm, yshift=0.2cm]{$b,\frac{1}{4},2$} (q2)
		
		(q1) edge	[below, out=-160,in=-20] node {$a,1,3$} (q0)
		
		(q0) edge [above, out=20,in=160] node[yshift=-0.07cm] {$a,\frac{1}{2},2$} (q1)
		(q1) edge [above] node [yshift=-0.1cm] {$b,2,5$} (q2)
		(q0) edge	[below, out=-50,in=-150] node [xshift=1.3cm, yshift=0.4cm]{$b,\frac{3}{2},4$} (q2)
		;
	\end{tikzpicture}
	\caption[An NMDA example.]{\label{fig:NMDAExample}An NMDA $\A$. The labeling on the transitions indicate the alphabet letter, the weight of the transition, and its discount factor.}
\end{figure}
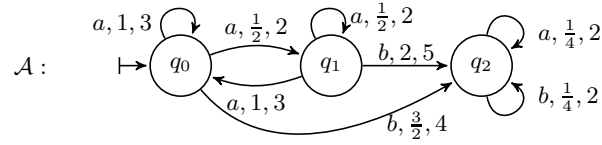

\Subject{Counter machines}
A two-counter machine \cite{Min67} $\M$ is a sequence $(l_1,\ldots,l_n)$ of commands, for some $n\in\Nat$, involving two counters $x$ and $y$. We refer to
$\set{1,\ldots,n}$ as the {\em locations} of the machine. For every $i\in\set{1,\ldots,n}$ we refer to $l_i$ as the {\em command in location $i$}. There are five possible forms of commands:
$$\inc(c),\ \dec(c),\ \goto l_k,\  \jz{c}{l_k}{l_{k'}},\  \halt,$$
where $c\in \set{x,y}$ is a counter and $1\le k,k'\le n$ are locations. 
For not decreasing a zero-valued counter $c\in\set{x,y}$, every $\dec(c)$ command is preceded by the command  $\jz{c}{\text{<current\_line>}}{\text{<next\_line>}}$, and there are no other direct goto-commands to it. 
The counters are initially set to $0$.
An example of a two-counter machine is given in \cref{fig:machineExample}.
\begin{figure}
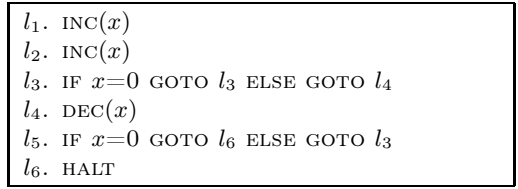

	\vspace*{-\baselineskip}
	\centering
	\setlength{\belowcaptionskip}{-\baselineskip}
	\fbox{\begin{minipage}{20em}
			\begin{enumerate}[itemsep=0pt]
				\item[$l_1$.] $\inc(x)$
				\item[$l_2$.] $\inc(x)$
				\item[$l_3$.] $\jz{x}{l_3}{l_4}$
				\item[$l_4$.] $\dec(x)$
				\item[$l_5$.] $\jz{x}{l_6}{l_3}$
				\item[$l_6$.] $\halt$
			\end{enumerate}
	\end{minipage}}
	\caption{\label{fig:machineExample}An example of a two-counter machine.}
\end{figure}

Let $L$ be the set of possible commands in $\M$, then a {\em run} of $\M$ is a sequence
$\CMrun=\CMrun_1,\ldots,\CMrun_m\in (L\times\Nat\times\Nat)^*$ such that the following hold:
\begin{enumerate}
	\item $\CMrun_1=\tuple{l_1,0,0}$.
	\item For all $1< i\le m$, let $\CMrun_{i-1}=(l_j,\alpha_x,\alpha_y)$ and $\CMrun_{i}=(l',\alpha_x',\alpha_y')$. Then, the following hold.
	\begin{itemize}
		\item If $l_j$ is an $\inc(x)$ command (resp. $\inc(y)$), then $\alpha_x'=\alpha_x+1$, $\alpha_y'=\alpha_y$ (resp. $\alpha_y=\alpha_y+1$, $\alpha_x'=\alpha_x$), and $l'=l_{j+1}$.
		\item If $l_j$ is $\dec(x)$ (resp. $\dec(y)$) then $\alpha_x'=\alpha_x-1$, $\alpha_y'=\alpha_y$ (resp. $\alpha_y=\alpha_y-1$, $\alpha_x'=\alpha_x$), and $l'=l_{j+1}$.
		\item If $l_j$ is $\goto l_k$ then $\alpha_x'=\alpha_x$, $\alpha_y'=\alpha_y$, and $l'=l_k$.
		\item If $l_j$ is $\jz{x}{l_k}{l_{k'}}$ then $\alpha_x'=\alpha_x$, $\alpha_y'=\alpha_y$, and $l'=l_k$ if $\alpha_x=0$, and $l'=l_{k'}$ otherwise.
		\item If $l_j$ is $\jz{y}{l_k}{l_{k'}}$ then $\alpha_x'=\alpha_x$, $\alpha_y'=\alpha_y$, and $l'=l_k$ if $\alpha_y=0$, and $l'=l_{k'}$ otherwise.		
		\item If $l'$ is $\halt$ then $i=m$, namely a run does not continue after $\halt$.
	\end{itemize}
\end{enumerate}
If, in addition, we have that $\CMrun_m=\tuple{l_j,\alpha_x,\alpha_y}$ such that $l_j$ is a $\halt$ command, we say that $\CMrun$ is a \emph{halting run}. We say that a machine $\M$ 0-halts if its run is halting and ends in $\tuple{l,0,0}$.
We say that a sequence of commands $\tau\in L^*$ {\em fits} a run $\CMrun$, if $\tau$ is the projection of $\CMrun$ on its first component.

The {\em command trace} $\pi=\sigma_1,\ldots,\sigma_{m}$ of a halting run $\CMrun=\CMrun_1,\ldots,\CMrun_m$ describes the flow of the run, including a description of whether a counter $c$ was equal to $0$ or larger than $0$ in each occurrence of an $\jz{c}{l_k}{l_{k'}}$ command. It is formally defined as follows.
$\sigma_{m}=\halt$ and for every $1< i\le m$, we define $\sigma_{i-1}$ according to $\CMrun_{i-1}=(l_j,\alpha_x,\alpha_y)$ in the following manner:
\begin{itemize}
	\item $\sigma_{i-1}=l_j$ if $l_j$ is not of the form $\jz{c}{l_k}{l_{k'}}$.
	\item 
	$\sigma_{i-1}=(\goto l_k,c=0)$ for $c\in\{x,y\}$, if $\alpha_c=0$ and the command $l_j$ is of the form $\jz{c}{l_k}{l_{k'}}$.
	\item 
	$\sigma_{i-1}=(\goto l_{k'},c>0)$ for $c\in\{x,y\}$, if $\alpha_c>0$ and the command $l_j$ is of the form $\jz{c}{l_k}{l_{k'}}$.
\end{itemize}
For example, the command trace of the halting run of the machine in \cref{fig:machineExample} is $\inc(x)$, $\inc(x)$, $(\goto l_4, x>0)$, $\dec(x)$, $(\goto l_3, x>0)$, $(\goto l_4, x>0)$, $\dec(x)$, $(\goto l_6, x=0)$, $\halt $.

Deciding whether a given counter machine $\M$ halts is known to be undecidable \cite{Min67}. Deciding whether $\M$ halts with both counters having value $0$, termed the {\em $0$-halting problem}, is also undecidable. 
Indeed, the halting problem can be reduced to the latter by adding some commands that clear the counters, before every \halt command. 

\section{Comparison of NMDAs}

We show that comparison of (integral) NMDAs is undecidable by reduction from the halting problem of two-counter machines. Notice that our NMDAs only use integral discount factors, while they do have non-integral weights. Yet,  weights can be easily changed to integers as well, by multiplying them all by a common denominator and making the corresponding adjustments in the calculations.

We start with a lemma on the accumulated value of certain series of discount factors and weights. Observe that by the lemma, no matter where the pair of discount-factor $\lambda\in\Nat\setminus\{0,1\}$ and weight $w=\frac{\lambda-1}{\lambda}$ appear along the run, they will have the same effect on the accumulated value. This property will play a key role in simulating counting by NMDAs.
\begin{lemma}\label{lem:undecidabilityContainment}
	For every sequence $\lambda_1,\cdots,\lambda_{m}$ of integers larger than $1$ and weights $w_1,\cdots,w_{m}$ such that $w_i=\frac{\lambda_i-1}{\lambda_i}$, we have
	$\sum_{i=1}^{m}{ \big(w_i \cdot \prod_{j=1}^{i-1} \frac{1}{\lambda_j}\big)}=
	1-\frac{1}{\prod_{j=1}^{m}\lambda_j}
	$.
\end{lemma}
\iftoggle{completeProofs}{
\begin{proof}
	We show the claim by induction on $m$.
	
	\noindent
	The base case, i.e.\ $m=1$, is trivial.
	For the induction step we have
	\begin{align*}
		\sum_{i=1}^{m+1}{ \big(w_i \cdot \prod_{j=1}^{i-1} \frac{1}{\lambda_j}\big)}&=
		\sum_{i=1}^{m}{ \big(w_i \cdot \prod_{j=1}^{i-1} \frac{1}{\lambda_j}\big)} + w_{m+1} \cdot \prod_{j=1}^{m} \frac{1}{\lambda_j}\\
		&= 1-\frac{1}{\prod_{j=1}^{m}\lambda_j} + \frac{\lambda_{m+1}-1}{\lambda_{m+1}} \cdot \prod_{j=1}^{m} \frac{1}{\lambda_j}\\
		&= 1-\frac{\lambda_{m+1}}{\prod_{j=1}^{m+1}\lambda_j} + \frac{\lambda_{m+1}-1}{\prod_{j=1}^{m+1}\lambda_j}
		= 1-\frac{1}{\prod_{j=1}^{m+1}\lambda_j}
	\end{align*}
\end{proof}
}
{The proof is by induction on $m$ and appears in \cite{BH23}.}

\subsection{The Reduction}\label{sec:TheReduction}
We turn to our reduction from the halting problem of two-counter machines to the problem of NMDA containment. We provide the construction and the correctness lemma with respect to automata on finite words, and then show in \cref{sec:Undecidability} how to use the same construction also for automata on infinite words.

Given a two-counter machine $\M$ with the commands $(l_1,\ldots,l_n)$,
we construct an integral DMDA $\A$ and an integral NMDA $\B$ on finite words, such that $\M$ $0$-halts iff there exists a word $w\in\Sigma^+$ such that $\B(w)\geq \A(w)$ iff there exists a word $w\in\Sigma^+$ such that $\B(w) > \A(w)$.

The automata $\A$ and $\B$ operate over the following alphabet $\Sigma$, which consists of   $5n+5$ letters, standing for the possible elements in a command trace of $\M$:
\begin{align*}
	\incdec = \ &\set{\inc(x),\dec(x),\inc(y),\dec(y)} \\ 
	\allgoto =\ &\big\{\goto\ l_k: k\in \{1,\ldots,n\}\big\}\cup\\
	&\big\{(\goto\ l_k,c=0): k\in \{1,\ldots,n\},c\in\{x,y\}\big\}\cup\\
	&\big\{(\goto\ l_{k'},c>0): k'\in \{1,\ldots,n\},c\in\{x,y\}\big\} \\
	\withouthalt =\ &\incdec \cup \allgoto \\
	\Sigma =\ &\withouthalt \cup \big\{\halt\big\}
\end{align*}

When $\A$ and $\B$ read a word $w\in\Sigma^+$, they intuitively simulate a sequence of commands $\tau_u$ that induces the command trace $u=\pref_{\halt}(w)$. 
If $\tau_u$ fits the actual run of $\M$, and this run 0-halts, then the minimal run of $\B$ on $w$ has a value strictly larger than $\A(w)$. 
If, however, $\tau_u$ does not fit the actual run of $\M$, or it does fit the actual run but it does not 0-halt, then the violation is detected by $\B$, which has a run on $w$ with value strictly smaller than $\A(w)$.

In the construction, we use the following partial discount-factor functions $\rho_p,\rho_d:\withouthalt\to \Nat$ and partial weight functions $\gamma_p,\gamma_d:\withouthalt\to \Rat$. 
$$
\rho_p(\sigma)=\begin{cases}
	5 & \sigma=\inc(x)\\
	4 & \sigma=\dec(x)\\
	7 & \sigma=\inc(y)\\
	6 & \sigma=\dec(y)\\
	15 & \text{otherwise}
\end{cases} ~~~~~
\rho_d(\sigma)=\begin{cases}
	4 & \sigma=\inc(x)\\
	5 & \sigma=\dec(x)\\
	6 & \sigma=\inc(y)\\
	7 & \sigma=\dec(y)\\
	15 & \text{otherwise}
\end{cases}
$$
$\gamma_p(\sigma)=\frac{\rho_p(\sigma)-1}{\rho_p(\sigma)}$, and $\gamma_d(\sigma)=\frac{\rho_d(\sigma)-1}{\rho_d(\sigma)}$.
We say that $\rho_p$ and $\gamma_p$ are the \emph{primal} discount-factor and weight functions, while $\rho_d$ and $\gamma_d$ are the \emph{dual} functions.
Observe that for every $c\in\{x,y\}$ we have that 
\begin{align}
	\rho_p(\inc(c))=\rho_d(\dec(c))>\rho_p(\dec(c))=\rho_d(\inc(c)) \label{eqn:primalDual}
\end{align}

Intuitively, we will use the primal functions for $\A$'s discount factors and weights, and the dual functions for identifying violations.
Notice that if changing the primal functions to the dual ones in more occurrences of $\inc(c)$ letters than of $\dec(c)$ letters  along some run, then by \cref{lem:undecidabilityContainment} the run will get a value lower than the original one.

We continue with their formal definitions.
$\A=\tuple{\Sigma,\{q_\A,q_\A^h\},\{q_\A\},\delta_\A,\gamma_\A,\rho_\A}$ is an integral DMDA consisting of two states, as depicted in \cref{fig:undecidabilityContainment_A}.
Observe that the initial state $q_\A$ has self loops for every alphabet letter in $\withouthalt$ with weights and discount factors according to the primal functions, and a transition $(q_\A,\halt, q_\A^h)$ with weight of $\frac{14}{15}$ and a discount factor of $15$.

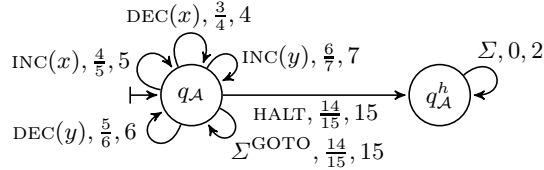
\begin{figure}
	\vspace*{-\baselineskip}
	\centering
	\setlength{\belowcaptionskip}{-\baselineskip}
	\begin{tikzpicture}[->,>=stealth',shorten >=1pt,auto,node distance=2cm, semithick, initial text=, every initial by arrow/.style={|->}]
		\node[initial, state] (q0) {$q_\A$};
		\node[state] (q1) [right of=q0, xshift=1.3cm] {$q_{\A}^h$};
		\path
		(q0) edge [loop left, in=125, out=170, looseness=5] node {$\inc(x),\frac{4}{5},5$} (q0)
		(q0) edge [loop above, in=70, out=115, looseness=5] node [yshift=-0.07cm]{$\dec(x),\frac{3}{4},4$} (q0)
		(q0) edge [loop right, in=25, out=60, looseness=5] node {$\inc(y),\frac{6}{7},7$} (q0)
		(q0) edge [loop right, in=-60, out=-25, looseness=5] node [yshift=-0.25cm, xshift=-0.15cm] {$\allgoto,\frac{14}{15},15$} (q0)
		(q0) edge [loop left, in=-150, out=-110, looseness=5] node [xshift=-0.1cm, yshift=0.1cm] {$\dec(y),\frac{5}{6},6$} (q0)
		
		(q0) edge [below] node [yshift=0.05cm] {$\halt,\frac{14}{15},15$} (q1)
		(q1) edge [loop above, in=0, out=40, looseness=5] node [xshift = 0.2cm, yshift=0.05cm]{$\Sigma,0,2$} (q1)
		;
	\end{tikzpicture}
	\caption{\label{fig:undecidabilityContainment_A}The DMDA $\A$ constructed for the proof of \cref{cl:undecidabilityContainment}.}
\end{figure}

The integral NMDA $\B=\tuple{\Sigma,Q_\B,\iota_\B,\delta_\B,\gamma_\B,\rho_\B}$ is the union of the following eight gadgets (checkers), each responsible for checking a certain type of violation in the description of a 0-halting run of $\M$.
It also has the states $\qfr,\qhalt\in Q_\B$ such that for all $\sigma\in \Sigma$, there are 0-weighted transitions $(\qfr,\sigma,\qfr)\in\delta_\B$ and $(\qhalt,\sigma,\qhalt)\in\delta_\B$ with an arbitrary discount factor.
Observer that in all of $\B$'s gadgets, the transition over the letter \halt to $\qhalt$ has a weight higher than the weight of the corresponding transition in $\A$, so that when no violation is detected, the value of $\B$ on a word is higher than the value of $\A$ on it.

\checker{1. Halt}{}
This gadget, depicted in \cref{fig:haltChecker}, checks for violations of non-halting runs.
Observe that its initial state $\qhc$ has self loops identical to those of $\A$'s initial state, a transition to $\qhalt$ over \halt with a weight higher than the corresponding weight in $\A$, and a transition to the state $q_{\mathsf{last}}$ over every letter that is not \halt, ``guessing'' that the run ends without a \halt command.

\begin{figure}
	\vspace*{-\baselineskip}
	\centering
	\setlength{\belowcaptionskip}{-\baselineskip}
	\begin{tikzpicture}[->,>=stealth',shorten >=1pt,auto,node distance=2cm, semithick, initial text=, every initial by arrow/.style={|->},
		every state/.style={ 
			inner sep=0pt}]
		\node[initial, state] (q0) {$\qhc$};
		\node[state] (q1) [right of=q0, xshift=1.3cm] {$\qhalt$};
		
		\node[state] (q2) [below of=q1, xshift=-1.2cm,yshift=0.5cm] {$q_{\mathsf{last}}$};
		\node[state] (q3) [right of=q2, xshift=0.5cm] {$\qfr$};
		
		\path
		(q0) edge [loop left, in=125, out=170, looseness=5] node {$\inc(x),\frac{4}{5},5$} (q0)
		(q0) edge [loop above, in=70, out=115, looseness=5] node {$\dec(x),\frac{3}{4},4$} (q0)
		(q0) edge [loop right, in=25, out=60, looseness=5] node {$\inc(y),\frac{6}{7},7$} (q0)
		(q0) edge [loop below, in=-95, out=-55, looseness=5] node [align=center] {\footnotesize$\allgoto$,\\$\frac{14}{15},15$} (q0)
		(q0) edge [loop left, in=-150, out=-110, looseness=5] node [xshift=-0.1cm, yshift=0.1cm] {$\dec(y),\frac{5}{6},6$} (q0)
		
		(q0) edge [below] node [yshift=0.05cm]{\footnotesize$\halt$,$\frac{15}{16},16$} (q1)
		(q1) edge [loop above, in=0, out=40, looseness=5] node [xshift = 0.2cm, yshift=0.05cm]{\footnotesize$\Sigma,0,2$} (q1)
		
		(q0) edge [right] node [xshift = 0.3cm, yshift=-0.1cm] {\footnotesize$\withouthalt,0,2$} (q2)
		(q2) edge [below] node {\footnotesize$\Sigma,2,2$} (q3)
		(q3) edge [loop above, in=0, out=40, looseness=5] node [xshift = 0.2cm, yshift=0.05cm]{\footnotesize$\Sigma,0,2$} (q3)
		;
	\end{tikzpicture}
	\caption{\label{fig:haltChecker}The Halt Checker in the NMDA $\B$.}
\end{figure}
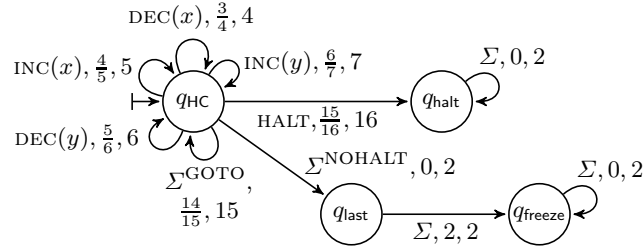

\checker{2. Negative-Counters}{}
The second gadget, depicted in \cref{fig:NegativeCounter}, checks that the input prefix $u$ has no more $\dec(c)$ than $\inc(c)$ commands for each counter $c\in\{x,y\}$.
It is similar to $\A$, however having self loops in its initial states that favor $\dec(c)$ commands when compared to $\A$. 

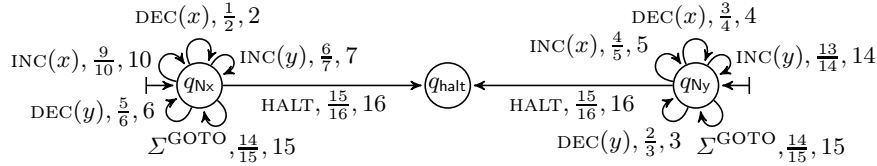
\begin{figure}
	\centering
	\setlength{\belowcaptionskip}{-\baselineskip}
	\begin{tikzpicture}[->,>=stealth',shorten >=1pt,auto,node distance=2cm, semithick, initial text=, every initial by arrow/.style={|->}, every state/.style={inner sep=0pt, minimum size=0.6cm}]
		\node[initial, state] (q0) {$q_{\mathsf{Nx}}$};
		\node[state] (q1) [right of=q0, xshift=1.3cm] {$\qhalt$};
		\path
		(q0) edge [loop left, in=125, out=170, looseness=5] node {$\inc(x),\frac{9}{10},10$} (q0)
		(q0) edge [loop above, in=70, out=115, looseness=5] node {$\dec(x),\frac{1}{2},2$} (q0)
		(q0) edge [loop right, in=25, out=60, looseness=5] node {$\inc(y),\frac{6}{7},7$} (q0)
		(q0) edge [loop below, in=-80, out=-40, looseness=5] node {\footnotesize$\allgoto$,$\frac{14}{15},15$} (q0)
		(q0) edge [loop left, in=-150, out=-110, looseness=5] node [xshift=-0.1cm, yshift=0.1cm] {$\dec(y),\frac{5}{6},6$} (q0)
		
		(q0) edge [below] node [yshift=0.05cm] {$\halt,\frac{15}{16},16$} (q1)
		;
		
		\node[initial right, state] (q0) [right of=q1, xshift=1.3cm]{$q_{\mathsf{Ny}}$};
		\path
		(q0) edge [loop left, in=125, out=170, looseness=5] node [yshift=0.2cm]{$\inc(x),\frac{4}{5},5$} (q0)
		(q0) edge [loop above, in=70, out=115, looseness=5] node {$\dec(x),\frac{3}{4},4$} (q0)
		(q0) edge [loop right, in=25, out=60, looseness=5] node {$\inc(y),\frac{13}{14},14$} (q0)
		(q0) edge [loop below, in=-80, out=-40, looseness=5] node [xshift=0.7cm] {\footnotesize$\allgoto$,$\frac{14}{15},15$} (q0)
		(q0) edge [loop left, in=-150, out=-110, looseness=5] node [xshift=0.3cm, yshift=-0.3cm] {$\dec(y),\frac{2}{3},3$} (q0)
		
		(q0) edge [below] node [yshift=0.05cm] {$\halt,\frac{15}{16},16$} (q1)
		;
				
	\end{tikzpicture}
	\caption{\label{fig:NegativeCounter}The negative-counters checker, on the left for $x$ and on the right for $y$, in the NMDA $\B$.}
\end{figure}

\checker{3. Positive-Counters}{}
The third gadget, depicted in \cref{fig:balancedCounters}, checks that for every $c\in\{x,y\}$, the input prefix $u$ has no more $\inc(c)$ than $\dec(c)$ commands.
It is similar to $\A$, while having self loops in its initial state according to the dual functions rather than the primal ones. 

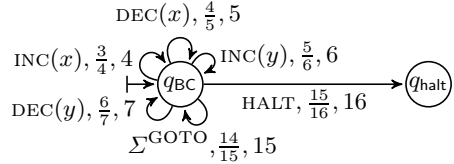
\begin{figure}
	\vspace*{-\baselineskip}
	\centering
	\setlength{\belowcaptionskip}{-\baselineskip}
	\begin{tikzpicture}[->,>=stealth',shorten >=1pt,auto,node distance=2cm, semithick, initial text=, every initial by arrow/.style={|->}, every state/.style={inner sep=0pt, minimum size=0.6cm}]
		\node[initial, state] (q0) {$q_{\mathsf{BC}}$};
		\node[state] (q1) [right of=q0, xshift=1.3cm] {$\qhalt$};
		\path
		(q0) edge [loop left, in=125, out=170, looseness=5] node {$\inc(x),\frac{3}{4},4$} (q0)
		(q0) edge [loop above, in=70, out=115, looseness=5] node {$\dec(x),\frac{4}{5},5$} (q0)
		(q0) edge [loop right, in=25, out=60, looseness=5] node {$\inc(y),\frac{5}{6},6$} (q0)
		(q0) edge [loop below, in=-80, out=-40, looseness=5] node {\footnotesize$\allgoto$,$\frac{14}{15},15$} (q0)
		(q0) edge [loop left, in=-150, out=-110, looseness=5] node [xshift=-0.1cm, yshift=0.1cm] {$\dec(y),\frac{6}{7},7$} (q0)
		
		(q0) edge [below] node [yshift=0.05cm] {$\halt,\frac{15}{16},16$} (q1)
		;
	\end{tikzpicture}
	\caption{\label{fig:balancedCounters}The Positive-Counters Checker in the NMDA $\B$.}
\end{figure}

\checker{4. Command}{}
The next gadget checks for local violations of successive commands. That is, it makes sure that the letter $w_i$ represents a command that can follow the command represented by $w_{i-1}$ in $\M$, ignoring the counter values. 
For example, if the command in location $l_2$ is $\inc(x)$, then from state $q_2$, which is associated with $l_2$, we move with the letter $\inc(x)$ to $q_3$, which is associated with $l_3$.
The test is local, as this gadget does not check for violations involving illegal jumps due to the values of the counters. 
An example of the command checker for the counter machine in \cref{fig:machineExample} is given in \cref{fig:commandCheckerExample}.

\begin{figure}
	\vspace*{-\baselineskip}
	\centering
	\setlength{\belowcaptionskip}{-\baselineskip}
	\begin{tikzpicture}[->,>=stealth',shorten >=1pt,auto,node distance=2.3cm, semithick, initial text=, every initial by arrow/.style={|->}, every state/.style={ 
			inner sep=0pt, minimum size=0.4cm}]
		\node[initial above, state] (q1) {$q_1$};
		\node[state] (q2) [right of=q1] {$q_2$};
		\node[state] (q3) [right of=q2] {$q_3$};
		\node[state] (q4) [right of=q3] {$q_4$};
		\node[state] (q5) [right of=q4] {$q_5$};
		\node[state] (q6) [right of=q5] {$q_6$};
		\node[state] (hlt) [below of=q6, yshift=1cm] {$\qhalt$};
		\node[state] (fr) [below of=q3, yshift=1cm] {$\qfr$};
		\path
		(q1) edge [above] node [align=center]{\footnotesize$\inc(x),\frac{4}{5},5$} (q2)
		(q2) edge [above] node [align=center, xshift=-0.1cm]{\footnotesize$\inc(x),\frac{4}{5},5$} (q3)
		(q3) edge [loop above,out=80,in=130, looseness=11] node [align=center, xshift=-1.2cm]{\footnotesize{$(\goto l_3,x=0),\frac{14}{15},15$}} (q3)
		(q3) edge [above] node [align=right, xshift=0.15cm]{\footnotesize{$\mbox{\sc goto } l_4$}\\$x>0,\frac{14}{15},15$} (q4)
		(q4) edge [above] node [align=center]{\footnotesize$\dec(x),\frac{3}{4},4$} (q5)
		(q5) edge [above] node [align=center]{\footnotesize$(\goto l_6,x=0),$\\$\frac{14}{15},15$} (q6)
		(q5) edge [above,out=120,in=60] node [align=center, xshift=0.4cm]{\footnotesize$(\goto l_3,x>0),\frac{14}{15},15$} (q3)
		(q6) edge [left] node [align=center]{\footnotesize$\halt,$\\$\frac{15}{16},16$} (hlt)
		(q1) edge [below] node [align=center,xshift=-0.7cm]{\footnotesize$\Sigma\setminus\{\inc(x)\},$\\\footnotesize$0,2$} (fr)
		(q2) edge [below] node [align=center,xshift=-0.7cm]{} (fr)
		(q3) edge [below] node [align=center,xshift=-0.7cm]{} (fr)
		(q4) edge [below] node [align=center,xshift=-0.7cm]{} (fr)
		(q5) edge [below] node [align=center,xshift=-0.7cm]{} (fr)
		(q6) edge [below] node [align=center,xshift=0.7cm]{\footnotesize$\Sigma\setminus\{\halt\},$\\\footnotesize$0,2$} (fr)
		;
	\end{tikzpicture}
	\caption{\label{fig:commandCheckerExample}The command checker that corresponds to the counter machine in \cref{fig:machineExample}.}
\end{figure}
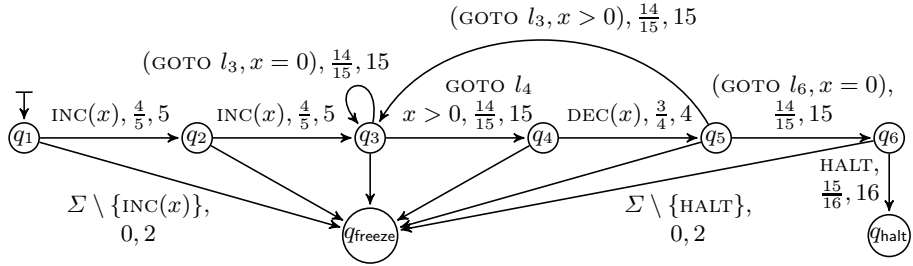

The command checker, which is a DMDA, consists of states $q_1,\ldots,q_n$ that correspond to the commands $l_1,\ldots,l_n$, and the states $\qhalt$ and $\qfr$. 
For two locations $j$ and $k$, there is a transition from $q_j$ to $q_k$ on the letter $\sigma$ iff $l_k$ can {\em locally follow\/} $l_j$ in a run of $\M$ that has $\sigma$ in the corresponding location of the command trace. 
That is, either $l_j$ is a $\goto l_k$ command (meaning $l_j=\sigma=\goto l_k$), $k$ is the next location after $j$ and $l_j$ is an $\inc$ or a $\dec$ command (meaning $k=j+1$ and $l_j=\sigma\in\incdec$), $l_j$ is an $\jz{c}{l_k}{l_{k'}}$ command with $\sigma=(\goto l_k,c=0)$, or $l_j$ is an $\jz{c}{l_s}{l_k}$ command with $\sigma=(\goto l_k,c>0)$.
The weights and discount factors of the $\withouthalt$ transitions mentioned above are according to the primal functions $\gamma_p$ and $\rho_p$ respectively.
For every location $j$ such that $l_j=\halt$, there is a transition from $q_j$ to $\qhalt$ labeled by the letter $\halt$ with a weight of $\frac{15}{16}$ and a discount factor of $16$.
Every other transition that was not specified above leads to $\qfr$ with weight $0$ and some discount factor.

\checker{5,6. Zero-Jump}{s}
	The next gadgets, depicted in \cref{fig:zeroJumpChecker}, check for violations in conditional jumps. In this case, we use a different checker instance for each counter $c\in\{x,y\}$, ensuring that for every $\jz{c}{l_k}{l_{k'}}$ command,  if the jump $\goto l_k$ is taken, then the value of $c$ is indeed $0$.

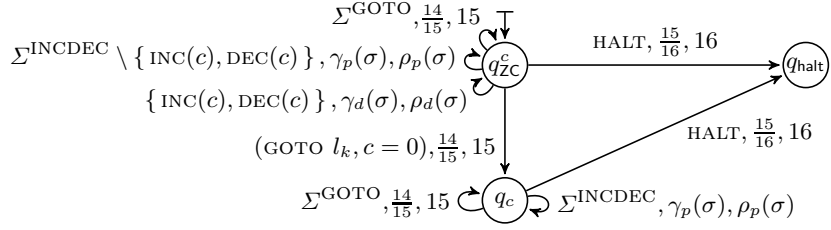
\begin{figure}
	\vspace*{-\baselineskip}
	\centering
	\setlength{\belowcaptionskip}{-\baselineskip}
	\begin{tikzpicture}[->,>=stealth',shorten >=1pt,auto,node distance=4cm, semithick, initial text=, every initial by arrow/.style={|->}, every state/.style={inner sep=0pt, minimum size=0.6cm}]
		\node[initial above, state] (q0) {$\qzc$};
		\node[state] (q1) [below of=q0, yshift=2.2cm] {$q_c$};
		\node[state] (hlt) [right of=q0] {$\qhalt$};
		\path
		(q0) edge [loop left, in=110, out=140, looseness=5] node [align=center, yshift=0.2cm, xshift=0.1cm] {\footnotesize$\allgoto$,$\frac{14}{15},15$} (q0)
		(q0) edge [loop left, in=160, out=-170, looseness=5] node [yshift=0.07cm] {$\incdec\setminus\set{\inc(c), \dec(c)},\gamma_p(\sigma),\rho_p(\sigma)$} (q0)
		(q0) edge [loop left, in=-150, out=-120, looseness=5] node [yshift=-0.1cm] {$\set{\inc(c), \dec(c)},\gamma_d(\sigma),\rho_d(\sigma)$} (q0)
		(q0) edge [left] node [align=center,yshift=-0.2cm] {\footnotesize$(\goto l_k,c=0),$\footnotesize$\frac{14}{15},15$} (q1)
		
		(q1) edge [loop right, in=-20, out=10, looseness=6] node {$\incdec,\gamma_p(\sigma),\rho_p(\sigma)$} (q1)
		
		(q1) edge [loop left, in=165, out=-160, looseness=6] node [align=center] {\footnotesize$\allgoto$,$\frac{14}{15},15$} (q1)
		
		(q0) edge [above] node [align=center] {$\halt,\frac{15}{16},16$} (hlt)
		
		(q1) edge [right] node [align=center, xshift=0.3cm] {$\halt,\frac{15}{16},16$} (hlt)
		;
	\end{tikzpicture}
	\caption{\label{fig:zeroJumpChecker}The Zero-Jump Checker (for a counter $c\in\set{x,y}$) in the NMDA $\B$.}
\end{figure}

Intuitively, $\qzc$ profits from words that have more $\inc(c)$ than $\dec(c)$ letters, while $q_c$ continues like $\A$.
If the move to $q_c$ occurred after a balanced number of $\inc(c)$ and $\dec(c)$, as it should be in a real command trace, neither the prefix word before the move to $q_c$, nor the suffix word after it result in a profit. 
Otherwise, provided that the counter is $0$ at the end of the run (as guaranteed by the negative- and positive-counters checkers), both prefix and suffix words get profits, resulting in a smaller value for the run. 

\checker{7,8. Positive-Jump}{s}
These gadgets, depicted in \cref{fig:positiveJumpChecker}, are dual to the zero-jump checkers, checking for the dual violations in conditional jumps. 
Similarly to the zero-jump checkers, we have a different instance for each counter $c\in\{x,y\}$, ensuring that for every $\jz{c}{l_k}{l_{k'}}$ command, if the jump $\goto l_{k'}$ is taken, then the value of $c$ is indeed greater than $0$.

\begin{figure}
	\centering
	\begin{tikzpicture}[->,>=stealth',shorten >=1pt,auto,node distance=2.5cm, semithick, initial text=, every initial by arrow/.style={|->},
		every state/.style={ 
			inner sep=0pt, minimum size=0.8cm}]
		\node[initial above, state] (q0) {$\qpc{0}$};
		\node[state] (q1) [below of=q0, yshift=0.2cm] {$\qpc{1}$};
		\node[state] (q2) [below of=q1, yshift=0.8cm] {$\qpc{2}$};
		\node[state] (qfr) [left of=q1,yshift=1cm] {$\qfr$};
		\node[state] (hlt) [right of=q1,xshift=2cm,yshift=0.5cm] {$\qhalt$};
		\path
		(q0) 
		edge 
		[loop left, out=150, in=120, looseness=4] 
		node 
		[align=center]
		{\footnotesize$\allgoto,\frac{14}{15},15$} 
		(q0)
		
		(q0) 
		edge 
		[loop right, out=60, in=30, looseness=4] 
		node 
		[align=center] {\footnotesize$\incdec\setminus\set{\inc(c)},\gamma_p(\sigma),\rho_p(\sigma)$}
		(q0)
		
		(q0) 
		edge 
		[right] 
		node 
		[align=left] 
		{$\inc(c),$\\\footnotesize$\gamma_d(\inc(c)),$\\$\rho_d(\inc(c))$}
		(q1)
		
		(q0) 
		edge 
		[above right] 
		node 
		[align=center, xshift=-0.5cm, yshift=0.1cm] 
		{\footnotesize$\halt,\frac{15}{16},16$}
		(hlt)
		
		(q0) 
		edge 
		[above left] 
		node 
		[align=center] 
		{\footnotesize$(\goto l_{k'},c>0),0,2$} 
		(qfr)
		
		(q1) 
		edge 
		[loop right, out=-30, in=-60, looseness=4] 
		node 
		[align=left,yshift=0.4cm] 
		{\footnotesize$\incdec,$\\\footnotesize$\gamma_p(\sigma),\rho_p(\sigma)$} 
		(q1)
		
		(q1) 
		edge 
		[loop left, out=-150, in=-180, looseness=4] 
		node 
		[align=center] 
		{\footnotesize$\allgoto,\frac{14}{15},15$} 
		(q1)
		
		(q1) 
		edge 
		[left] 
		node 
		[align=center] 
		{\footnotesize$(\goto l_{k'},c>0),\frac{14}{15},15$} 
		(q2)
		
		(q2) 
		edge 
		[loop right, out=0, in=-30, looseness=5] 
		node 
		[align=center] 
		{\footnotesize$\incdec\setminus\set{\inc(c),\dec(c)},\gamma_p(\sigma),\rho_p(\sigma)$} 
		(q2)
		
		(q2) 
		edge 
		[loop left, out=-140, in=-165, looseness=5] 
		node 
		[align=center, yshift=-0.1cm] 
		{\footnotesize$\allgoto,\frac{14}{15},15$} 
		(q2)
		
		(q2) 
		edge 
		[loop left, out=180, in=155, looseness=5] 
		node 
		[align=center, yshift=0.1cm]
		{$\set{\inc(c),\dec(c)},\gamma_d(\sigma),\rho_d(\sigma)$} 
		(q2)
		
		(q2) 
		edge 
		[below right] 
		node 
		[align=center] 
		{\footnotesize$\halt,\frac{15}{16},16$} 
		(hlt)
		
		(q1) 
		edge 
		[above] 
		node 
		[align=center,xshift=0.3cm] 
		{\footnotesize$\halt,1,2$} 
		(qfr)
		;
	\end{tikzpicture}
	\caption{\label{fig:positiveJumpChecker}The Positive-Jump Checker (for a counter $c$) in the NMDA $\B$.}
\end{figure}
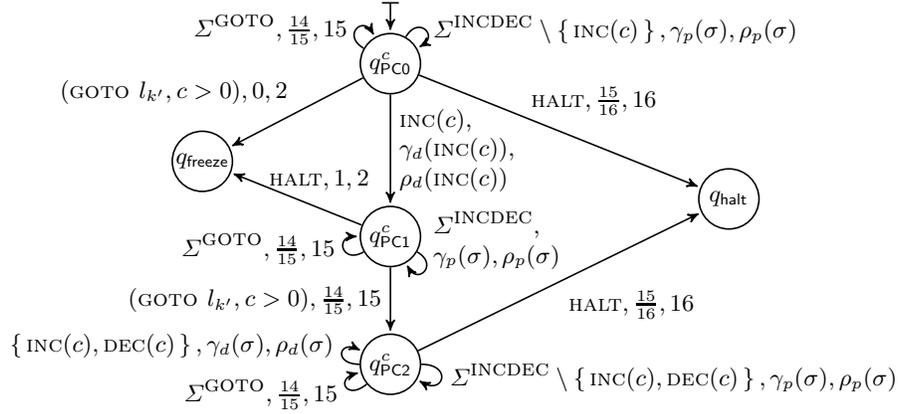

Intuitively, if the counter is $0$ on a $(\goto l_{k'},c>0)$ command when there was no $\inc(c)$ command yet, the gadget benefits by moving from $\qpc{0}$ to $\qfr$. If there was an $\inc(c)$ command, it benefits by having the dual functions on the move from $\qpc{0}$ to $\qpc{1}$ over $\inc(c)$ and the primal functions on one additional self loop of $\qpc{1}$ over $\dec(c)$.

\begin{lemma}
	\label{cl:undecidabilityContainment}
	Given a two-counter machine $\M$, we can compute an integral DMDA $\A$ and an integral NMDA $\B$ on finite words, such that $\M$ $0$-halts iff there exists a word $w\in\Sigma^+$ such that $\B(w)\geq \A(w)$ iff there exists a word $w\in\Sigma^+$ such that $\B(w) > \A(w)$.
\end{lemma}
\iftoggle{completeProofs}{
\begin{proof}
Given a two-counter machine $\M$, consider the DMDA $\A$ and the NMDA $\B$ constructed in \cref{sec:TheReduction}, and an input word $w$. Let $u=\pref_{\halt}(w)$.

We prove the claim by showing that I) if $u$ correctly describes a 0-halting run of $\M$ then $\B(w)>\A(w)$, and II) if $u$ does not fit the actual run of $\M$, or if it does fit it, but the run does not 0-halt, then the violation is detected by $\B$, in the sense that $\B(w)<\A(w)$.

\vspace{5pt}\noindent{\bf I.} 	
We start with the case that $u$ correctly describes a 0-halting run of $\M$, and show that $\B(w)>\A(w)$.

Observe that in all of $\B$'s checkers, the transition over the \halt command to the $\qhalt$ state has a weight higher than the weight of the corresponding transition in $\A$. Thus, if a checker behaves like $\A$ over $u$, namely uses the primal functions, it generates a value higher than that of $\A$.

We show below that each of the checkers generates a value higher than the value of $\A$ on $u$ (which is also the value of $\A$ on $w$), also if it nondeterministically ``guesses a violation'', behaving differently than $\A$.

\lchecker{1. Halt}{}
Since $u$ does have the \halt command, the run of the halt checker on $u$, if guessing a violation, will end in the pair of transitions from $\qhc$ to $q_{\mathsf{last}}$ to $\qfr$ with discount factor $2$ and weights $0$ and $2$, respectively. 

Let $D$ be the accumulated discount factor in the gadget up to these pair of transitions. 
According to \cref{lem:undecidabilityContainment}, the accumulated weight at this point is $1-\frac{1}{D}$, hence the value of the run will be $1-\frac{1}{D} + \frac{1}{D}\cdot 0 + \frac{1}{2D}\cdot 2 = 1$, which is, according to \cref{lem:undecidabilityContainment}, larger than the value of $\A$ on any word.

\lchecker{2,3. Negative- and Positive-Counters}{s}
Since $u$ has the same number of $\inc(c)$ and $\dec(c)$ letters, by \cref{eqn:primalDual} and $\cref{lem:undecidabilityContainment}$, these gadgets and $\A$ will have the same value on the prefix of $u$ until the last transition, on which the gadgets will have a higher weight.

\lchecker{4. Command}{}
As this gadget is deterministic, it cannot ``guess a violation'', and its value on $u$ is larger than $\A(u)$ due to the weight on the \halt command.

\lchecker{5,6. Zero-Jump}{s}
Consider a counter $c\in\set{x,y}$ and a run $r$ of the gadget on $u$.
If $r$ did not move to $q_c$, we have $\B(r)>\A(w)$, similarly to the analysis in the negative- and positive-counters checkers.
Otherwise, denote the transition that $r$ used to move to $q_c$ as $t$.
Observe that since $u$ correlates to the actual run of $\M$, we have that $t$ was indeed taken when $c=0$.
In this case the value of the run will not be affected, since before $t$ we have the same number of $\inc(c)$ and $\dec(c)$ letters, and after $t$ we also have the same number of $\inc(c)$ and $\dec(c)$ letters. Hence, due to the last transition over the \halt command, we have $\B(r)>\A(u)$.

\lchecker{7,8. Positive-Jump}{s}
Consider a counter $c\in\set{x,y}$ and a run $r$ of the gadget on $u$.
If $r$ never reaches $\qpc{1}$, it has the same sequence of weights and discount factors as $\A$, except for the higher-valued \halt transition.
If $r$ reaches $\qpc{1}$ but never reaches $\qpc{2}$, since $u$ ends with a $\halt$ letter, we have that $r$ ends with a transition to $\qfr$ that has a weight of $1$, hence $\B(r)=1>\A(w)$.

If $r$ reaches $\qpc{2}$, let $u=y\con\inc(c)\con z\con v$ where $y$ has no $\inc(c)$ letters, $t=r[|y|+1+|z|]$ is the first transition in $r$ targeted at $\qpc{2}$, and $\alpha_c\geq 1$ is the value of the counter $c$ when $t$ is taken.
We have that $1+\numof{\inc(c)}{z}=\numof{\dec(c)}{z}+\alpha_c$. 
Since $u$ is balanced, we also have that $\numof{\dec(c)}{v}=\numof{\inc(c)}{v}+\alpha_c$.
For the first $\inc(c)$ letter, $r$ gets a discount factor of $\rho_d(\inc(c))=\rho_p(\dec(c))$.
All the following $\inc(c)$ and $\dec(c)$ letters contribute discount factors according to $\rho_p$ in $z$ and according to $\rho_d$ in $v$.
Hence, $r$ gets the discount factor $\rho_p(\dec(c))$ a total of 
\begin{align*}
	1+\numof{\dec(c)}{z}+\numof{\inc(c)}{v}&=
	1+1+\numof{\inc(c)}{z}-\alpha_c+\numof{\inc(c)}{v}\\
	&=
	\numof{\inc(c)}{u}+1-\alpha_c\\&\leq \numof{\inc(c)}{u}=\numof{\dec(c)}{u}
\end{align*}
times, and the discount factor $\rho_p(\inc(c))$ a total of 
\begin{align*}
	\numof{\inc(c)}{z}+\numof{\dec(c)}{v}&=
	\numof{\inc(c)}{z}+\numof{\inc(c)}{v}+\alpha_c\\&=
	\numof{\inc(c)}{u}-1+\alpha_c\geq \numof{\inc(c)}{u}
\end{align*}
times.

Therefore, the value of $r$ is at least as big as the value of $\A$ on the prefix of $u$ until the \halt transition, and due to the higher weight of $r$ on the latter, we have $\B(r) > \A(u)$.

\vspace{5pt}\noindent{\bf II.} 	
We continue with the case that $u$ does not correctly describe a 0-halting run of $\M$, and show that $\B(w) <\A(w)$.	
Observe that the incorrectness must fall into one of the following cases, each of which results in a lower value of one of $\B$'s gadgets on $u$, compared to the value of $\A$ on $u$:
\begin{itemize}
	\item {\it The word $u$ has no \halt command.} In this case the minimal-valued run of the halt checker on $u$ will be the same as of $\A$ until the last transition, on which the halt checker will have a $0$ weight, compared to a strictly positive weight in $\A$.
	
	\item {\it The word $u$ does not describe a run that ends up with value $0$ in both counters.}
	Then there are the following sub-cases:
	\begin{itemize}
		\item {\it The word $u$ has more $\dec(c)$ than $\inc(c)$ letters for some counter $c\in\{x,y\}$}.
		For $c=x$, in the negative-counters checker, more discount factors were changed from $4$ to $2$  than those changed from $5$ to $10$, compared to their values in $\A$, implying that the total value of the gadget until the last letter will be lower than of $\A$ on it.
		For $c=y$, we have a similar analysis with respect to the discount factors $6;3$, and $7;14$.
		\item {\it The word $u$ has more $\inc(c)$ than $\dec(c)$ letters for some counter $c\in\{x,y\}$}.
		By \cref{eqn:primalDual} and $\cref{lem:undecidabilityContainment}$, the value of the positive-counters checker until the last transition will be lower than of $\A$ until the last transition.
	\end{itemize}
	
	Observe, though, that the weight of the gadgets on the \halt transition ($16$) is still higher than that of $\A$ on it ($15$). 
	Nevertheless, since a ``violation detection'' results in replacing at least one discount factor from $4$ to $2$, from $6$ to $3$, from $5$ to $4$, or from $7$ to $6$ (and replacing the corresponding weights, for preserving the $\frac{\rho-1}{\rho}$ ratio), and the ratio difference between $16$ and $15$ is less significant than between the other pairs of weights, we have that the gadget's value and therefore $\B$'s value on $u$ is smaller than $\A(u)$.
	Indeed, by \cref{lem:undecidabilityContainment} $\A(u)=1-\frac{1}{D_\A}$, where $D_\A$ is the multiplication of the discount factors along $\A$'s run, and $\B(u)\leq 1-(\frac{1}{D_\A}\cdot \frac{7}{6}\cdot\frac{15}{16}) < 1-\frac{1}{D_\A} = \A(u)$.
	
	\item {\it The word $u$ does not correctly describe the run of $\M$}. Then there are the following sub-cases:
	\begin{itemize}
		\item {\it The incorrect description does not relate to conditional jumps}. 
		Then the command-checker has the same weights and discount factors as $\A$ on the prefix of $u$ until the incorrect description, after which it has $0$ weights, compared to strictly positive weights in $\A$.
		\item {\it The incorrect description relates to conditional jumps}. Then there are the following sub-sub-cases:
		\begin{itemize}
			\item {\it A counter $c>0$ at a position $i$ of $\M$'s run, while $u[i]=\goto l_k,c=0$}. Let $v=u[0..i{-}1]$ and $u=v\con v'$, and consider the run $r$ of the zero-jump checker on $u$ that moves to $q_c$ after $v$.
			Then $\numof{\inc(c)}{v}>\numof{\dec(c)}{v}$ and $\numof{\inc(c)}{v'}<\numof{\dec(c)}{v'}$. (We may assume that the total number of $\inc(c)$ and $\dec(c)$ letters is the same, as otherwise one of the previous checkers detects it.)
			
			All the $\inc(c)$ and $\dec(c)$ transitions in $r[0..i{-}1]$ have weights and discount factors according to the dual functions, and those transitions in $r[i..|w|{-}1]$ have weights and discount factors according to the primal functions.
			Therefore, compared to $\A$, more weights changed from $\gamma_p(\inc(c))$ to $\gamma_d(\inc(c))=\gamma_p(\dec(c))$ than weights changed from $\gamma_p(\dec(c))$ to $\gamma_d(\dec(c))=\gamma_p(\inc(c))$, resulting in a lower  total value of $r$ than of $\A$ on $u$. (As shown for the negative- and positive-counters checkers, the higher weight of the $\halt$ transition is less significant than the lower values above.)
			
			\item {\it A counter $c=0$ at a position $i$ of $\M$'s run, while $u[i]=\goto l_k,c>0$}.
			Let $r$ be a minimal-valued run of the positive-jump checker on $u$.
			
			If  there are no $\inc(c)$ letters in $u$ before position $i$, $r$ will have the same weights and discount factors as $\A$ until the $i$'s letter, on which it will move from $\qpc{1}$ to $\qfr$, continuing with $0$-weight transitions, compared to strictly positive ones in $\A$.
			
			Otherwise, we have that the first $\inc(c)$ letter of $u$ takes $r$ from $\qpc{0}$ to $\qpc{1}$ with a discount factor of $\rho_d(\inc(c))$. 
			Then in $\qpc{1}$ we have more $\dec(c)$ transitions than $\inc(c)$ transitions, and in $\qpc{2}$ we have the same number of $\dec(c)$ and $\inc(c)$ transitions. (We may assume that $u$ passed the previous checkers, and thus has the same total number of $\inc(c)$ and $\dec(c)$ letters.)
			Hence, we get two more discount factors of $\rho_d(\inc(c))$ than $\rho_p(\inc(c))$, resulting in a value smaller than $\A(u)$. (As in the previous cases, the higher value of the \halt transition is less significant.)
		\end{itemize}
	\end{itemize}
\end{itemize}
\end{proof}}
{The proof uses the construction presented above, and can be found in \cite{BH23}.}

\subsection{Undecidability of Comparison}\label{sec:Undecidability}
For finite words, the undecidability result directly follows from \cref{cl:undecidabilityContainment} and the undecidability of the 0-halting problem of counter machines \cite{Min67}.
\begin{theorem}\label{cl:ContainmentFiniteWordsUndecidable}
	Strict and non-strict containment of (integral) NMDAs on finite words are undecidable. More precisely, the problems of deciding for given integral NMDA $\N$ and integral DMDA $\D$ whether $\N(w) \leq \D(w)$ for all finite words $w$ and whether $\N(w) < \D(w)$ for all finite words $w$.
\end{theorem}

For infinite words, undecidability of non-strict containment also follows from the reduction given in \cref{sec:TheReduction}, as the reduction considers prefixes of the word until the first \halt command. 
We leave open the question of whether strict containment is also undecidable for infinite words. The problem with the latter is that a \halt command might never appear in an infinite word $w$ that incorrectly describes a halting run of the two-counter machine, in which case both automata $\A$ and $\B$ of the reduction will have the same value on $w$. 
On words $w$ that have a \halt command but do not correctly describe a halting run of the two-counter machine we have $\B(w)<\A(w)$, and on a word $w$ that does correctly describe a halting run we have $\B(w)>\A(w)$. Hence, the reduction only relates to whether $\B(w)\leq\A(w)$ for all words $w$, but not to whether $\B(w) <\A(w)$ for all words $w$.

\begin{theorem}\label{cl:NonStrictContainmentInfiniteWordsUndecidable}
	Non-strict containment of (integral) NMDAs on infinite words is undecidable. More precisely, the problem of deciding for given integral NMDA $\N$ and integral DMDA $\D$ whether $\N(w) \leq \D(w)$ for all infinite words $w$.
\end{theorem}
\begin{proof}
	The automata $\A$ and $\B$ in the reduction given in \cref{sec:TheReduction} can operate as is on infinite words, ignoring the Halt-Checker gadget of $\B$ which is only relevant to finite words.

	Since the values of both $\A$ and $\B$ on an input word $w$ only relate to the prefix $u=\pref_{\halt(w)}$ of $w$ until the first \halt command, we still have that $\B(w)>\A(w)$ if $u$ correctly describes a halting run of the two-counter machine $\M$ and that $\B(w)<\A(w)$ if $u$ is finite and does not correctly describe a halting run of $\M$.
	
	Yet, for infinite words there is also the possibility that the word $w$ does not contain the \halt command. In this case, the value of both $\A$ and the command checker of $\B$ will converge to $1$, getting $\A(w)=\B(w)$.
	
	Hence, if $\M$ 0-halts, there is a word $w$, such that $\B(w)>\A(w)$ and otherwise, for all words $w$, we have $\B(w)\leq \A(w)$. 
\end{proof}

Observe that for NMDAs, equivalence and non-strict containment are interreducible.

\begin{theorem}\label{cl:EquivalenceFiniteWordsUndecidable}
	Equivalence of (integral) NMDAs on finite as well as infinite words is undecidable. That is, the problem of deciding for given integral NMDAs $\A$ and $\B$ on finite or infinite words whether $\A(w) = \B(w)$ for all words $w$.
\end{theorem}
\begin{proof}
	Assume toward contradiction the existence of a procedure for equivalence check of $\A$ and $\B$. 
	We can use the nondeterminism to obtain an automaton $\C=\A \cup \B$, having $C(w)\leq A(w)$ for all words $w$. We can then check whether $\C$ is equivalent to $\A$, which holds if and only if $\A(w) \leq \B(w)$ for all words $w$. Indeed, if $\A(w) \leq \B(w)$ then $\A(w) \leq \min(\A(w), \B(w)) = \C(w)$, while if there exists a word $w$, such that $\B(w)<\A(w)$, we have $\C(w) = \min(\A(w), \B(w))  <\A(w)$, implying that $\C$ and $\A$ are not equivalent. Thus, such a procedure contradicts the undecidability of non-strict containment, shown in \cref{cl:ContainmentFiniteWordsUndecidable,cl:NonStrictContainmentInfiniteWordsUndecidable}.
\end{proof}

\section{Comparison of NDAs with Different Discount Factors}\label{sec:TwoNdas}

We present below our algorithm for the comparison of NDAs with different discount factors.
We start with automata on infinite words, and then show how to solve the case of finite words by reduction to the case of infinite words.

The algorithm is based on our main observation that, due to the difference between the discount factors, we only need to consider the combination of the automata computation trees up to some level $k$, after which we can consider first the best/worst continuation of the automaton with the smaller discount factor, and on top of it the worst/best continuation of the second automaton.

For an NDA $\A$, we define its \emph{lowest} (resp.\ \emph{highest}) \emph{infinite run value} by 
		$\lowest(\A)$ (resp.\ $\highest(\A)$) = $\min$ (resp.\ $\max$) $\{\A(r) \ST r$ is an infinite run of  $\A~(\text{on some word }w \in \Sigma^\omega)\}$.
		
Observe that we can use $\min$ and $\max$ (rather than $\inf$ and $\sup$) since the infimum and supremum values are indeed attainable by specific infinite runs of the NDA (cf.\ \cite[Proof of Theorem 9]{BL21}).
Notice that $\lowest(\A)$ and $\highest(\A)$ can be calculated in PTIME by a simple reduction to one-player discounted-payoff games \cite{Andersson06}.  

Considering word values, we also refer to the \emph{lowest} (resp.\ \emph{highest}) \emph{word value} of $\A$, defined by $\lowestw(\A)$ (resp.\ $\highestw(\A)$)= $\min$ (resp.\ $\max$) $\set{\A(w) \ST w \in \Sigma^\omega}$.
Observe that $\lowestw(\A)=\lowest(\A)$, $\highestw(\A)\leq \highest(\A)$, and for deterministic automaton, $\highestw(\A)=\highest(\A)$.

For an NMDA $\A$ with states $Q$, we define the \emph{maximal difference between suffix runs} of $\A$ as 
$
\maxdiff{\A}=\max\set{\highest(\A^q)-\lowest(\A^q)\ST {q\in Q} }
$. 
Notice that $\maxdiff{\A}\geq 0$ and that $\A^q(w)$ is bounded as follows.
\begin{equation}\label{eq:ValuesRange}
	\lowest(\A^q) \leq \A^q(w) \leq \lowest(\A^q) + \maxdiff{\A}
\end{equation}

\begin{lemma}\label{cl:NDAsContainmentInfinite}
	There is an algorithm that computes for every input discount factors $\lambda_A,\lambda_D\in\Rat\cap (1,\infty)$, $\lambda_A$-NDA $\A$ and $\lambda_D$-DDA $\D$ on infinite words the value of $\min\{\A(w)-\D(w) \ST w \in \Sigma^\omega\}$.
\end{lemma}
\begin{proof}
	Consider an alphabet $\Sigma$, discount factors $\lambda_A,\lambda_D\in\Rat\cap (1,\infty)$, a $\lambda_A$-NDA $\A=\tuple{\Sigma, Q_\A, \iota_\A, \delta_\A, \gamma_\A}$ and a $\lambda_D$-DDA $\D=\tuple{\Sigma, Q_\D, \iota_\D, \delta_\D, \gamma_\D}$.
	When $\lambda_A=\lambda_D$, we can generate a $\lambda_A$-NDA $\C\equiv \A-\D$ over the product of $\A$ and $\D$ and compute $\lowestw(\C)$.
	
	When $\lambda_A\neq\lambda_D$, {\bf we consider first the case that $\lambda_A<\lambda_D$}.
	
	Our algorithm unfolds the computation trees of $\A$ and $\D$, up to a level in which only the minimal-valued suffix words of $\A$ remain relevant --	
	Due to the massive difference between the accumulated discount factor in $\A$ compared to the one in $\D$, any ``penalty'' of not continuing with a minimal-valued suffix word in $\A$, defined below as $m_\A$, cannot be compensated even by the maximal-valued word of $\D$, which ``profit'' is at most as high as $\maxdiff{\D}$.
	Hence, at that level, it is enough to look among the minimal-valued suffixes of $\A$ for the one that implies the highest value in $\D$.
	
	For every transition $t=(q,\sigma,q')\in\delta_\A$, let $\minval(q,\sigma,q')=\gamma_\A(q,\sigma,q') + \frac{1}{\lambda_A}\cdot\lowestw(\A^{q'})$ be the best (minimal) value that $\A^q$ can get by taking $t$ as the first transition.
	We say that $t$ is \emph{preferred} if it starts a minimal-valued infinite run of $\A^q$, namely 
	$\preferred = \set{t=(q,\sigma,q')\in\delta_\A \ST \minval(t)=\lowestw(\A^q)}$ is the set of preferred transitions of $\A$.
	Observe that an infinite run of $\A^q$ that takes only transitions from $\preferred$, has a value equal to $\lowest(\A^q)$ (cf.\ \cite[Proof of Theorem 9]{BL21}).
	
	If all the transitions of $\A$ are preferred, $\A$ has the same value on all words, and then $\min\{\A(w)-\D(w) \ST w \in \Sigma^\omega\}=\lowest(\A)-\highestw(\D)$. (Recall that since $\D$ is deterministic, we can easily compute $\highestw(\D)$.)
	Otherwise, let $m_\A$ be the minimal penalty for not taking a preferred transition in $\A$, meaning\
		$
		m_\A=\min\Big\{
		\minval(t')- \minval(t'') 
		\ \Big|\ 
		\begin{matrix}
			t'=(q,\sigma',q')\in\delta_\A\setminus\preferred,\\ t''=(q,\sigma'',q'')\in \preferred
		\end{matrix} \Big\}
		$.
		Observe that $m_\A > 0$.
		
		Considering the connection between $m_\A$ and $\maxdiff{\D}$, notice first that if $\maxdiff{\D}=0$, $\D$ has the same value on all words, and then we have $\min\{\A(w)-\D(w) \ST w \in \Sigma^\omega\}=\lowest(\A)-\lowest(\D)$.
		Otherwise, meaning $\maxdiff{\D} > 0$, we unfold the computation trees of $\A$ and $\D$ for the first $k$ levels, until the maximal difference between suffix runs in $\D$, divided by the accumulated discount factor of $\D$, is smaller than the minimal penalty for not taking a preferred transition in $\A$, divided by the accumulated discount factor of $\A$. Meaning, $k$ is the minimal integer such that
		\begin{equation}\label{eq:UnfoldingLevel}
		\frac{\maxdiff{\D}}{{\lambda_D}^k} < \frac{m_\A}{{\lambda_A}^k}
		\end{equation}
		Starting at level $k$, the penalty gained by taking a non-preferred transition of $\A$ cannot be compensated by a higher-valued word of $\D$.
		
		At level $k$, we consider separately every run $\psi$ of $\A$ on some prefix word $u$.
		We should look for a suffix word $w$, that minimizes 
		\begin{equation}\label{eq:Suffix}
			\A(uw)-\D(uw)=\A(\psi) + \frac{1}{{\lambda_A}^k} \cdot \A^{\delta_\A(\psi)}(w) - \D(u) - \frac{1}{{\lambda_D}^k} \cdot \D^{\delta_\D(u)}(w)
		\end{equation}
		
		A central point of the algorithm is that every word that minimizes $\A-\D$ must take only preferred transitions of $\A$ starting at level $k$ (\iftoggle{completeProofs}{see \cref{cl:NDAsContainmentInfiniteCorrectness}}{full proof in \cite{BH23}}).
		As all possible remaining continuations after level $k$ yield the same value in $\A$, we can choose among them the continuation that yields the highest value in $\D$.
		
		Let $\B$ be the partial automaton with the states of $\A$, but only its preferred transitions $\preferred$. (We ignore words on which $\B$ has no runs.)
		We shall use the automata product $\B^{\delta_\A(\psi)} \times \D^{\delta_\D(u)}$ to force suffix words that only take preferred transitions of $\A$, while calculating among them the highest value in $\D$.  
		
		Let $\C^{(\delta_\A(\psi),\delta_\D(u))}=\tuple{\Sigma, Q_\A \times Q_\D, \set{(\delta_\A(\psi), \delta_\D(u))}, \preferred \times \delta_\D, \gamma_\C}$ be the partial $\lambda_D$-NDA that is generated by the product of $\B^{\delta_\A(\psi)}$ and $\D^{\delta_\D(u)}$, while only considering the weights (and discount factor) of $\D$, meaning $\gamma_\C((q,p),\sigma,(q',p'))=\gamma_\D(p,\sigma,p')$.
		
		A word $w$ has a run in $\A^{\delta_\A(\psi)}$ that uses only preferred transitions iff $w$ has a run in $\C^{(\delta_\A(\psi),\delta_\D(u))}$.
		Also, observe that the nondeterminism in $\C$ is only related to the nondeterminism in $\A$, and the weight function of $\C$ only depends on the weights of $\D$, hence all the runs of $\C^{(\delta_\A(\psi),\delta_\D(u))}$ on the same word result in the same value, which is the value of that word in $\D$.
		Combining both observations, we get that a word $w$ has a run in $\A^{\delta_\A(\psi)}$ that uses only preferred transitions iff $w$ has a run $r$ in $\C^{(\delta_\A(\psi),\delta_\D(u))}$ such that $\C^{(\delta_\A(\psi),\delta_\D(u))}(r)=\D^{\delta_\D(u)}(w)$. 
		Hence, after taking the $k$-sized run $\psi$ of $\A$, and under the notations defined in \cref{eq:Suffix}, a suffix word $w$ that can take only preferred transitions of $\A$, and maximizes $\D^{\delta_\D(u)}(w)$, has a value of $\D^{\delta_\D(u)}(w)=\highest(\C^{(\delta_\A(\psi),\delta_\D(u))})$.
		This leads to
		\begin{align*}
			&\min\set{\A(v)-\D(v)\ST v\in\Sigma^\omega} =\\
			&\min\Big\{\A(\psi) + \frac{\A^{\delta_\A(\psi)}(w)}{{\lambda_A}^k} - \D(u) - \frac{\D^{\delta_\D(u)}(w)}{{\lambda_D}^k}\Big| \begin{matrix} 
			u\in\Sigma^k, w\in\Sigma^\omega, \\ 
			\psi \text{ is a run of } \A \text{ on } u
			\end{matrix} \Big\} =\\
			&\min_{\psi}\Bigg\{
			\A(\psi) + \frac{\lowest(\A^{\delta_\A(\psi)})}{{\lambda_A}^k} - \D(u) - \frac{\highest(\C^{(\delta_\A(\psi),\delta_\D(u))})}{{\lambda_D}^k}\Big|
			\begin{matrix} 
			u\in\Sigma^k, \\ \psi \text{ is a run}\\ 
			\text{of } \A \text{ on } u
			\end{matrix} \Bigg\}
		\end{align*}
		and it is only left to calculate this value for every $k$-sized run of $\A$, meaning for every leaf in the computation tree of $\A$. 

	\noindent {\bf The case of $\lambda_A>\lambda_D$} is analogous, with the following changes: 
	\begin{itemize}
		\item
		For every transition of $\D$, we compute $\maxval(p,\sigma,p')=\gamma_\D(p,\sigma,p') + \frac{1}{\lambda_D}\cdot\highestw(\D^{p'})$, instead of $\minval(q,\sigma,q')$.
		\item 
		The preferred transitions of $\D$ are the ones that start a maximal-valued infinite run, that is
		$\preferred = \set{t=(p,\sigma',p')\in\delta_\D \ST \maxval(t)=\highest(\D^p)}$, and the minimal penalty $m_\D$ is\\
		$
		m_\D=\min\Big\{
		\maxval(t'') - \maxval(t') 
		\ \Big|\ 
		\begin{matrix}
			t''=(p,\sigma'',p'')\in \preferred, \\
			t'=(p,\sigma',p')\in\delta_\D\setminus\preferred
		\end{matrix} \Big\}
		$
		
		\item 
		$k$ should be the minimal integer such that
		$
		\frac{\maxdiff{\A}}{{\lambda_A}^k} < \frac{m_\D}{{\lambda_D}^k}
		$.
		\item
		We define $\B$ to be the restriction of $\D$ to its preferred transitions, and $\C^{(\delta_\A(\psi),\delta_\D(u))}$ as a partial $\lambda_A$-NDA on the product of $\A^{\delta_\A(\psi)}$ and $\B^{\delta_\D(u)}$ while considering the weights of $\A$.
		We then calculate $\lowest(\C^{(\delta_\A(\psi),\delta_\D(u))})$ for every $k$-sized run of $\A$, $\psi$, and conclude that $\min\set{\A-\D}$ is equal to
		$
			\min_{\psi}\set{
				\A(\psi) + \frac{\lowest(\C^{(\delta_\A(\psi),\delta_\D(u))})}{{\lambda_A}^k} - \D(u) - \frac{\highest(\D_{\delta_\D(u)})}{{\lambda_D}^k}}
		$.
		
		Observe that in this case, it might not hold that all runs of $\C^{(\delta_\A(\psi),\delta_\D(u))}$ on the same word have the same value, but such property is not required, since we look for the minimal run value (which is the minimal word value).
	\end{itemize}
	
\end{proof}

Notice that the algorithm of \cref{cl:NDAsContainmentInfinite} does not work if switching the direction of containment, namely if considering a deterministic $\A$ and a nondeterministic $\D$. 
The determinism of $\D$ is required for finding the maximal value of a valid word in $\B^{\delta_\A(\psi)}\times \D^{\delta_\D(u)}$. If $\D$ is not deterministic, the maximal-valued run of $\B^{\delta_\A(\psi)}\times \D^{\delta_\D(u)}$ on some word $w$ equals the value of some run of $\D$ on $w$, but not necessarily the value of $\D$ on $w$.
We also need $\D$ to be deterministic for computing $\highestw(\D^{p})$ in the case that $\lambda_A>\lambda_D$.

\iftoggle{completeProofs}{
To show the correctness of \cref{cl:NDAsContainmentInfinite}, we present the following claim.
\begin{lemma}\label{cl:NDAsContainmentInfiniteCorrectness}
	For every input discount factors $\lambda_A,\lambda_D\in\Rat\cap (1,\infty)$ such that $\lambda_A<\lambda_D$, $\lambda_A$-NDA $\A$ and $\lambda_D$-DDA $\D$, every infinite word $w$ that minimizes $\A(w)-\D(w)$ must take a preferred transition of $\A$ at every level $k$ for which $\frac{\maxdiff{\D}}{{\lambda_D}^k} < \frac{m_\A}{{\lambda_A}^k}$.
\end{lemma}
\begin{proof}
	Consider discount factors $\lambda_A,\lambda_D\in\Rat\cap (1,\infty)$ such that $\lambda_A<\lambda_D$, $\lambda_A$-NDA $\A$, $\lambda_D$-DDA $\D$, and $k$ the minimal integer such that 
	$$\frac{\maxdiff{\D}}{{\lambda_D}^k} < \frac{m_\A}{{\lambda_A}^k}$$
	
	Assume toward contradiction the existence of a word $v$ that minimizes $\A-\D$, while a minimal-valued run $\psi_\A$ of $\A$ on $v$ does not take a preferred transition at some level $n\geq k$.
	Let $u$ be the $n$-sized prefix of $v$, 
	$w$ the corresponding suffix (meaning $v=u\con w$), 
	$\psi$ the prefix run of $\psi_\A$ on $u$, 
	and $w'$ some minimal-valued word of $\A^{\delta_\A(\psi)}$.
	The first transition taken by $\psi_\A$ when continuing with $w$ is not preferred, meaning \begin{equation}\label{eq:NotPreferredSuffix}
		\A^{\delta_\A(\psi)}(w) \geq \lowestw(\A^{\delta_\A(\psi)})+m_\A = \A^{\delta_\A(\psi)}(w') + m_\A
	\end{equation}
	Hence,
	\begin{align*}
		\A(v)-\D(v) &\overset{\mathrm{^{\eqref{eq:Suffix}}}}{=} 
		\A(\psi) + \frac{\A^{\delta_\A(\psi)}(w)}{{\lambda_A}^n}  - \D(u) - \frac{\D^{\delta_\D(u)}(w)}{{\lambda_D}^n} \\
		&\overset{\mathrm{^{\eqref{eq:NotPreferredSuffix},\eqref{eq:ValuesRange}}}}{\geq} \A(\psi) + \frac{\A^{\delta_\A(\psi)}(w')+m_\A}{{\lambda_A}^n}  - \D(u) - \frac{\lowest(\D^{\delta_\D(u)}) + \maxdiff{\D}}{{\lambda_D}^n} \\
		& \overset{\mathrm{^{\eqref{eq:UnfoldingLevel}}}}{>}
		\A(\psi) + \frac{\A^{\delta_\A(\psi)}(w')}{{\lambda_A}^n}  - \D(u) - \frac{\lowest(\D^{\delta_\D(u)})}{{\lambda_D}^n} \\
		&\overset{\mathrm{^{\eqref{eq:ValuesRange}}}}{\geq}
		\A(\psi) + \frac{\A^{\delta_\A(\psi)}(w')}{{\lambda_A}^n}  - \D(u) - \frac{\D^{\delta_\D(u)}(w')}{{\lambda_D}^n} \\
		&\overset{\mathrm{^{\eqref{eq:Suffix}}}}{=}\A(u\con w') - \D(u\con w')
	\end{align*}
	leading to a contradiction.
\end{proof}
}

Moving to automata on finite words, we reduce the problem to the corresponding problem handled in \cref{cl:NDAsContainmentInfinite}, by adding to the alphabet a new letter that represents the end of the word, and making some required adjustments.

\begin{lemma}\label{cl:NDAsContainmentFinite}
	There is an algorithm that computes for every input discount factors $\lambda_A,\lambda_D\in\Rat\cap (1,\infty)$, $\lambda_A$-NDA $\A$ and $\lambda_D$-DDA $\D$ on finite words the value of $\inf\set{\A(u)-\D(u) \ST u \in \Sigma^+}$, and determines if there exists a finite word $u$ for which $\A(u)-\D(u)$ equals that value.
\end{lemma}
\begin{proof}
	Without loss of generality, we assume that initial states of automata have no incoming transitions. (Every automaton can be changed in linear time to an equivalent automaton with this property.)

	We convert, as described below, an NDA $\N$ on finite words to an NDA $\hat{\N}$ on infinite words, such that $\hat{\N}$ intuitively simulates the finite runs of $\N$.
	For an alphabet $\Sigma$, a discount factor $\lambda\in\Rat\cap (1,\infty)$, and a $\lambda$-NDA (DDA) $\N=\tuple{\Sigma, Q_\N, \iota_\N, \delta_\N, \gamma_\N}$ on finite words, we define the $\lambda$-NDA (DDA)
	$\hat{\N} = \tuple{
		\hat{\Sigma}, 
		Q_\N \cup \set{q_\fin}, 
		\iota_\N, 
		\delta_{\hat{\N}}, 
		\gamma_{\hat{\N}}
	}$ on infinite words.
	The new alphabet $\hat{\Sigma}=\Sigma \cup \set{\fin}$ contains a new letter $\fin\notin\Sigma$ that indicates the end of a finite word.
	The new state $q_\fin$ has $0$-valued self loops on every letter in the alphabet, and there are $0$-valued transitions from every non-initial state to $q_\fin$ on the new letter $\fin$.
	Formally, 
	$\delta_{\hat{\N}}= \delta_\N \cup
		\set{(q_\fin,\sigma, q_\fin \ST \sigma \in \hat{\Sigma})} 
		\cup \set{(q,\fin, q_\fin \ST q\in Q_\N\setminus\iota_\N)} 
	$, and \\
	$\gamma_{\hat{\N}}(t) = \begin{cases}
		\gamma_\N(t) & t \in \delta_\N\\
		0 &  \text{otherwise}
	\end{cases}$
	
	\noindent Observe that for every state $q\in Q_\N$, the following hold.
	\begin{enumerate}
		\item 
		For every finite run $r_\N$ of $\N^q$, there is an infinite run $r_{\hat{\N}}$ of $\hat{\N}^q$, such that ${\hat{\N}}^q(r_{\hat{\N}}) = {{\N}}^q(r_{{\N}})$, and $r_{\hat{\N}}$ takes some $\fin$ transitions.
		($r_{\hat{\N}}$ can start as $r_\N$ and then continue with only $\fin$ transitions.)
		
		\item
		For every infinite run $r_{\hat{\N}}$ of $\hat{\N}^q$ that has a $\fin$ transition, there is a finite run $r_\N$ of $\N^q$, such that ${\hat{\N}}^q(r_{\hat{\N}}) = {{\N}}^q(r_{{\N}})$. 
		($r_\N$ can be the longest prefix of $r_{\hat{\N}}$ up to the first $\fin$ transition).
		
		\item
		For every infinite run $r_{\hat{\N}}$ of $\hat{\N}^q$ that has no $\fin$ transition, there is a series of finite runs of $\N^q$, such that the values of the runs in $\N^q$ converge to ${\hat{\N}}^q(r_{\hat{\N}})$. (For example, the series of all prefixes of $r_{\hat{\N}}$).
	\end{enumerate}
	Hence, for every $q\in Q_\N$ we have 
	$
		\inf\set{\N^q(r)\ST r \text{ is a run of }\N^q} =\lowest({\hat{\N}}^q)
	$ and
	$
		\sup\set{\N^q(r)\ST r \text{ is a run of }\N^q} =\highest({\hat{\N}}^q)
	$.
	(For a non-initial state $q$, we also consider the ``run'' of $\N^q$ on the empty word, and define its value to be $0$.)	
	Notice that the infimum (supremum) run value of $\N^q$ is attained by an actual run of $\N^q$ iff there is an infinite run of $\hat{\N}^q$ that gets this value and takes a $\fin$ transition.
	
	For every state $q\in Q_{\hat{\N}}$, we can determine, as follows, whether $\lowest({\hat{\N}}^q)$ is attained by an infinite run taking a $\fin$ transition.  
	We calculate $\lowest({\hat{\N}}^q)$ for all states, and then start a process that iteratively marks the states of $\hat{\N}$, such that at the end, $q\in Q_{\hat{\N}}$ is marked iff $\lowest({\hat{\N}}^q)$ can be achieved by a run with a $\fin$ transition.
	We start with $q_\fin$ as the only marked state.
	In each iteration we further mark every state $q$ from which there exists a preferred transition $t=(q,\sigma, q')\in \preferred$ to some marked state $q'$.
	The process terminates when an iteration has no new states to mark.
	Analogously, we can determine whether $\highest({\hat{\N}}^q)$ is attained by a run that goes to $q_\fin$.
	
	Consider discount factors $\lambda_A,\lambda_D\in\Rat\cap (1,\infty)$, a $\lambda_A$-NDA $\A$ and a $\lambda_D$-DDA $\D$ on finite words.
	When $\lambda_A=\lambda_D$, similarly to \cref{cl:NDAsContainmentInfinite}, the algorithm finds the infimum value of $\C\equiv\A-\D$ using $\hat{\C}$, and determines if an actual finite word attains this value using the process described above.
	
	Otherwise, the algorithm converts $\A$ and $\D$ to $\hat{\A}$ and $\hat{\D}$, and proceeds as in \cref{cl:NDAsContainmentInfinite} over $\hat{\A}$ and $\hat{\D}$.
	According to the above observations, we have that $\inf\set{\A(u)-\D(u)\ST u\in\Sigma^+} = \min\{\hat{\A}(w)-\hat{\D}(w) \ST w \in \Sigma^\omega\}$, and that $\inf\set{\A(u)-\D(u)}$ is attainable iff $\min\{\hat{\A}(w)-\hat{\D}(w)\}$ is attainable by some word that has a $\fin$ transition.
	Hence, whenever computing $\lowest$ or $\highest$, we also perform the process described above, to determine whether this value is attainable by a run that has a $\fin$ transition.
	We determine that $\inf\set{\A(u)-\D(u)}$ is attainable iff exists a leaf of the computation tree that leads to it, for which the relevant values $\lowest$ and $\highest$ are attainable.
\end{proof}

\subsubsection{Complexity analysis}
We show below that the algorithm of \cref{cl:NDAsContainmentInfinite,cl:NDAsContainmentFinite} only needs a polynomial space, with respect to the size of the input automata,  implying a PSPACE algorithm for the corresponding decision problems.
We define the size of an NDA $\N$, denoted by $|\N|$, as the maximum between the number of its transitions, the maximal binary representation of any weight in it, and the maximal unary representation of the discount factor. (Binary representation of the discount factors might cause our algorithm to use an exponential space, in case that the two factors are very close to each other.)
The input NDAs may have rational weights, yet it will be more convenient to consider equivalent NDAs with integral weights that are obtained by multiplying all the weights by their common denominator \cite{BH21}. (Observe that it causes the values of all words to be multiplied by this same ratio, and it keeps the same input size, up to a polynomial change.)

Before proceeding to the complexity analysis, we provide an auxiliary lemma\iftoggle{completeProofs}{}{ (proof appears in \cite{BH23})}.
\begin{lemma}\label{lem:lasso_run}
	For every integers $p>q \in \Nat\setminus\{0\}$, a $\frac{p}{q}$-NDA $\A$ with integral weights, and a lasso run $r=t_0, t_1, \ldots, t_{x{-}1}, (t_{x}, t_{x{+}1}, \ldots, t_{x+y-1})^\omega$ of $\A$, there exists an integer $b$, such that $\A(r) =\frac{b}{p^x(p^y-q^y)}$.
\end{lemma}
\iftoggle{completeProofs}{
\begin{proof}
Let $\lambda=\frac{p}{q}$ be $\A$'s discount factor, and $\gamma$ its weight function. Consider a lasso run $r=t_0, t_1, \ldots, t_{x{-}1}, (t_{x}, t_{x{+}1}, \ldots, t_{x+y-1})^\omega$ of $\A$.
Let $v_f=\gamma(t_0) + \frac{1}{\lambda} \gamma(t_1) + \ldots + \frac{1}{\lambda^{x{-}1}}\gamma(t_{x{-}1})$ be its prefix value, and $v_\ell=\gamma(t_x) + \frac{1}{\lambda} \gamma(t_{x{+}1}) + \ldots + \frac{1}{\lambda^{y{-}1}}\gamma(t_{x+y-1}) $ its loop value.

Since all the weights are integers, we have that $v_f=\frac{a_f}{p^{x}}$ and $v_\ell=\frac{a_\ell}{p^{y}}$ for some integers $a_f$ and $a_\ell$. 
Recall that for a loop $\ell$ of length $y$ and accumulated value $v_\ell$ in a $\lambda$-NDA, the accumulated value of its infinite repetition is $\sum_{i=0}^{\infty} \frac{v_\ell}{(\lambda^y)^ i} = v_\ell \frac{ \lambda^y}{\lambda^y-1}$.
Hence the value of $r$ is 
\begin{align*}
	\gamma(r) &= v_f + \frac{1}{\lambda^x}\cdot v_\ell\frac{\lambda^y}{\lambda^y-1}=
	\frac{a_f}{p^{x}} + \frac{a_\ell}{p^{y}}\cdot \frac{1}{\lambda^{x-y}(\lambda^y-1)}= \frac{a_f}{p^x} + \frac{a_\ell\cdot q^{x-y}}{p^{y+x-y}(\frac{p^y-q^y}{q^y})}\\
	&=
	\frac{a_f(p^y-q^y)+a_\ell\cdot q^x}{p^x(p^y-q^y)}
\end{align*}
\end{proof}
}

Proceeding to the complexity analysis, let the input size be $S=|\A|+|\D|$, 
the reduced forms of $\lambda_\A$ and $\lambda_\D$ be $\frac{{p}}{{q}}$ and $\frac{{p_\D}}{{q_\D}}$ respectively, the number of states in $\A$ be $n$, and the maximal difference between transition weights in $\D$ be $M$.
Observe that  
$n\leq S, p\leq S, M\leq 2\cdot 2^S$, $\frac{\lambda_\D}{\lambda_D - 1} \leq \frac{{p_\D}}{{p_\D}-{q_\D}} \leq {p_\D} \leq S$, and for $\lambda_{\D} > \lambda_{\A} > 1$, we also have $\frac{\lambda_\D}{\lambda_\A}=\frac{p\cdot q_\D}{q\cdot p_\D} \geq 1+\frac{1}{S^2}$.

Observe that $\A$ has a best infinite run (and $\D$ has a worst infinite run), in a lasso form as in \cref{lem:lasso_run}, with $x,y\in[1..n]$. 
Indeed, following preferred transitions, a run must complete a lasso, and then may forever repeat its choices of preferred transitions.
Hence, $m_\A$, being the difference between two lasso runs, is in the form of 
\begin{align*}
m_\A &= 
\frac{b_1}{{p}^{x_1}({p}^{y_1}-{q}^{y_1})} - \frac{b_2}{{p}^{x_2}({p}^{y_2}-{q}^{y_2})}
=
\frac{b_3}{{p}^{n}({p}^{y_1}-{q}^{y_1})({p}^{y_2}-{q}^{y_2})} > \frac{b_3}{{p}^{n}{p}^{y_1}{p}^{y_2}} \\&
\geq
\frac{1}{{p}^{3n}} 
\geq 
\frac{1}{{S}^{3S}} 
\stackrel{\text{for } S \geq 1}{>}
\frac{1}{({2^S})^{3S}}
=
\frac{1}{2^{3S^2}}
\end{align*}
for some $x_1, x_2, y_1, y_2\leq n$ and some integers $b_1, b_2, b_3$.
(Similarly, we can show that $m_\D > \frac{1}{2^{3S^2}}$.)
We have
$\maxdiff{\D}\leq M\cdot\frac{\lambda_{\D}}{\lambda_{\D}-1}$, hence
\begin{equation*}
\frac{\maxdiff{\D}}{m_\A} 
\leq 
\frac{M\cdot\frac{\lambda_{\D}}{\lambda_{\D}-1}}{m_\A}
\leq 
\frac{2^{1+S}\cdot S}{m_\A}
\stackrel{(\text{for } S \geq 1)}{<}
\frac{2^{3S}}{m_\A}
< 
2^{3S+3S^2}
\end{equation*}

Recall that we unfold the computation tree until level $k$, which is the minimal integer such that $(\frac{\lambda_D}{\lambda_A})^k > \frac{\maxdiff{\D}}{m_\A}$.
Observe that for $S\geq 1$ we have
$
\big(\frac{\lambda_{\D}}{\lambda_{\A}}\big)^{S^2} \geq \big(1 + \frac{1}{S^2}\big)^{S^2}\geq 2
$,
hence for $k'=S^2\cdot(3S+3S^2)$, we have
$$\big(\frac{\lambda_D}{\lambda_A}\big)^{k'} =
\big((\frac{\lambda_D}{\lambda_A})^{S^2}\big)^{3S+3S^2}
\geq 
2^ {3S+3S^2} > \frac{\maxdiff{\D}}{m_\A}$$
meaning that $k$ is polynomial in $S$.
Similar analysis shows that $k$ is polynomial in $S$ also for $\lambda_{\D} < \lambda_{\A}$.

Considering decision problems that use our algorithm, due to the equivalence of NPSPACE and PSPACE, the algorithm can nondeterministically guess an optimal prefix word $u$ of size $k$, letter by letter, as well as a run $\psi$ of $\A$ on $u$, transition by transition, and then compute the value of $\A(\psi) + \frac{\lowest(\A^{\delta_\A(\psi)})}{{\lambda_A}^k} - \D(u) - \frac{\highest(\C^{(\delta_\A(\psi),\delta_\D(u))})}{{\lambda_D}^k}$.

Observe that along the run of the algorithm, we need to save the following information, which can be done in polynomial space:
\begin{itemize}
	\item 
	The automaton $\C\equiv \B \times \D$ (or $\A \times \B$), which requires polynomial space.
	\item 
	${\lambda_\A}^k$ (for $\A(\psi)$) and ${\lambda_\D}^k$ (for $\D(u)$).
	Since we save them in binary representation, we have $\log_2 (\lambda^k)\leq k \log_2(S)$, requiring polynomial space.
\end{itemize}

We thus get the following complexity result.
\begin{theorem}\label{cl:TwoNdasDecision}
	For input discount factors $\lambda_A,\lambda_D\in\Rat\cap (1,\infty)$, $\lambda_A$-NDA $\A$ and $\lambda_D$-DDA $\D$ on finite or infinite words, it is decidable in PSPACE whether  $\A(w) \geq \D(w)$ and whether $\A(w) > \D(w)$ for all words $w$. 
\end{theorem}
\begin{proof}
	We use \cref{cl:NDAsContainmentInfinite} in the case of infinite words and \cref{cl:NDAsContainmentFinite} in the case of finite words, checking whether $\min\set{\A(w)-\D(w)} < 0$ and whether $\min\set{\A(w)-\D(w)} \leq 0$. In the case of finite words, we also use the information of whether there is an actual word that gets the desired value.
\end{proof}

Since integral NDAs can always be determinized \cite{BH14}, we get as a corollary that there is an algorithm to decide equivalence and strict and non-strict containment of integral NDAs with different (or the same) discount factors. Note, however, that it might not be in PSPACE, since determinization exponentially increases the number of states, resulting in $k$ that is exponential in $S$, and storing in binary representation values in the order of $\lambda^k$ might require exponential space. 
\begin{corollary}
	There are algorithms to decide for input integral discount factors $\lambda_A,\lambda_B\in\Nat$, $\lambda_A$-NDA $\A$ and $\lambda_B$-NDA $\B$ on finite or infinite words whether or not $\A(w) > \B(w)$, $\A(w) \geq \B(w)$, or $\A(w) = \B(w)$ for all words $w$.
\end{corollary}

\section{Conclusions}
The new decidability result, providing an algorithm for comparing discounted-sum automata with different integral discount factors, may allow to extend the usage of discounted-sum automata in formal verification, while the undecidability result strengthen the justification of restricting discounted-sum automata with multiple integral discount factors to tidy NMDAs.
The new algorithm also extends the possible, more limited, usage of discounted-sum automata with rational discount factors, while further research should be put into this direction.

\subsubsection{Acknowledgements} 
We thank Guillermo A. Perez for stimulating discussions on the comparison of integral NDAs with different discount factors.

%
%
%
\bibliographystyle{splncs04}
\bibliography{bibl}

\end{document}